\documentclass[10pt, two column, twoside]{IEEEtran}

\usepackage{amsmath}
\usepackage[ruled]{algorithm2e}
\usepackage{mathrsfs}
\usepackage{cite}
\usepackage{bm,epsfig,amsthm,url}
\usepackage{indentfirst}
\usepackage{amssymb}
\usepackage{amsfonts}
\usepackage{epstopdf}
\usepackage[caption=false,font=footnotesize]{subfig}
\usepackage{xcolor}

\newcommand{\bs}[1]{{\bm{#1}}}

\theoremstyle{definition}

\newtheorem{proposition}{Proposition}

\newtheorem{remark}{Remark}

\begin{document}

\title{Energy-Efficient Processing and Robust Wireless Cooperative Transmission for Edge Inference}

\author{Kai~Yang,~\IEEEmembership{Student Member,~IEEE,}
        Yuanming~Shi,~\IEEEmembership{Member,~IEEE,}
        Wei~Yu,~\IEEEmembership{Fellow,~IEEE}
        and~Zhi~Ding,~\IEEEmembership{Fellow,~IEEE}
\thanks{K. Yang is with the School of Information Science and Technology, ShanghaiTech University, Shanghai 201210, China, also with the Shanghai Institute of Microsystem and Information Technology, Chinese Academy of Sciences, Shanghai 200050, China, and also with the University of Chinese Academy of Sciences, Beijing 100049, China (e-mail: yangkai@shanghaitech.edu.cn).}
\thanks{Y. Shi is with the School of Information Science and Technology, ShanghaiTech University, Shanghai 201210, China (e-mail: shiym@shanghaitech.edu.cn).}
\thanks{Wei Yu is with the Electrical
and Computer Engineering Department, University of Toronto, Toronto,
ON M5S 3G4, Canada (e-mail: weiyu@comm.utoronto.ca).}
\thanks{Z. Ding is with the Department of Electrical and Computer Engineering,
University of California at Davis, Davis, CA 95616 USA (e-mail:
zding@ucdavis.edu).}
}

\maketitle
\IEEEpeerreviewmaketitle
\begin{abstract}
Edge machine learning can deliver low-latency and private artificial intelligent (AI) services for mobile devices by leveraging computation and storage resources at the network edge. This paper presents an energy-efficient edge processing framework to execute deep learning inference tasks at the edge computing nodes whose wireless connections to mobile devices are prone to channel uncertainties. Aimed at minimizing the sum of computation and transmission power consumption with probabilistic quality-of-service (QoS) constraints, we formulate a joint inference tasking and downlink beamforming problem that is characterized by a group sparse objective function. We provide a statistical learning based robust optimization approach to approximate the highly intractable probabilistic-QoS constraints by nonconvex quadratic constraints, which are further reformulated as matrix inequalities with a rank-one constraint via matrix lifting. We design a reweighted power minimization approach by iteratively reweighted $\ell_1$ minimization with difference-of-convex-functions (DC) regularization and updating weights, where the reweighted approach is adopted for enhancing group sparsity whereas the DC regularization is designed for inducing rank-one solutions. Numerical results demonstrate that the proposed approach outperforms other state-of-the-art approaches.
\end{abstract}


\section{Introduction}\label{introduction}
Machine learning has transformed many aspects of our daily lives by taking advantage of abundant data and computing power in the cloud center. In particular, the strong capability of capturing the representations of data for detection or classification using deep neural networks \cite{lecun2015deep} has made impressive gains in face recognition, natural language processing tasks, etc. With the explosion of mobile data and the increasing edge computing capability, there is an emerging trend of \textit{edge intelligence} \cite{zhou2019edge,park2018edgeai}. Instead of uploading all data collected by mobile devices to the remote cloud data center, edge intelligence emphasizes the use of the computation and storage resources at network edges to provide low-latency and reliable artificial intelligent (AI) service \cite{kang2019incentive,kang2019reliable} for privacy/security sensitive devices, such as wearable devices, augmented reality, smart vehicles, and drones. However, since mobile devices are usually equipped with limited computation power, storage and energy \cite{park2018edgeai}, it is usually infeasible to deploy deep learning models, i.e., deep neural networks (DNNs), at resource-constrained mobile devices, and execute inference tasks locally. A promising solution is to enable processing at the mobile network access points to facilitate deep learning inference, which is termed as \textit{edge inference} \cite{xu2018scaling,xu2019edgesanitizer}. 

In this paper, we shall present the edge processing framework for edge inference (as illustrated in Fig. \ref{fig:framework}) that the input (e.g., a piece of rough doodle) of each mobile user is uploaded to wireless access points (e.g., base stations) served as edge computing nodes, each task is performed with pre-trained deep learning model (e.g., Nvidia's AI system GauGAN \cite{GauGANwebsite2019} for turning rough doodles into photorealistic landscapes) at multiple edge computing nodes, and the output results (e.g., landscape images) are transmitted to mobile users via coordinated beamforming among multiple access points. In such a system, the provisioning of wireless transmissions in both the uplink and the downlink are important design considerations.
In addition to the low-latency requirement, improving the energy efficiency \cite{sze2017efficient} is also critical due to the high computational complexity of processing DNNs, for which a number of works focusing on model compression methods \cite{han2015deep,Zhang_SPM18}.

\begin{figure}[h]
  \centering
  \includegraphics[width=\columnwidth]{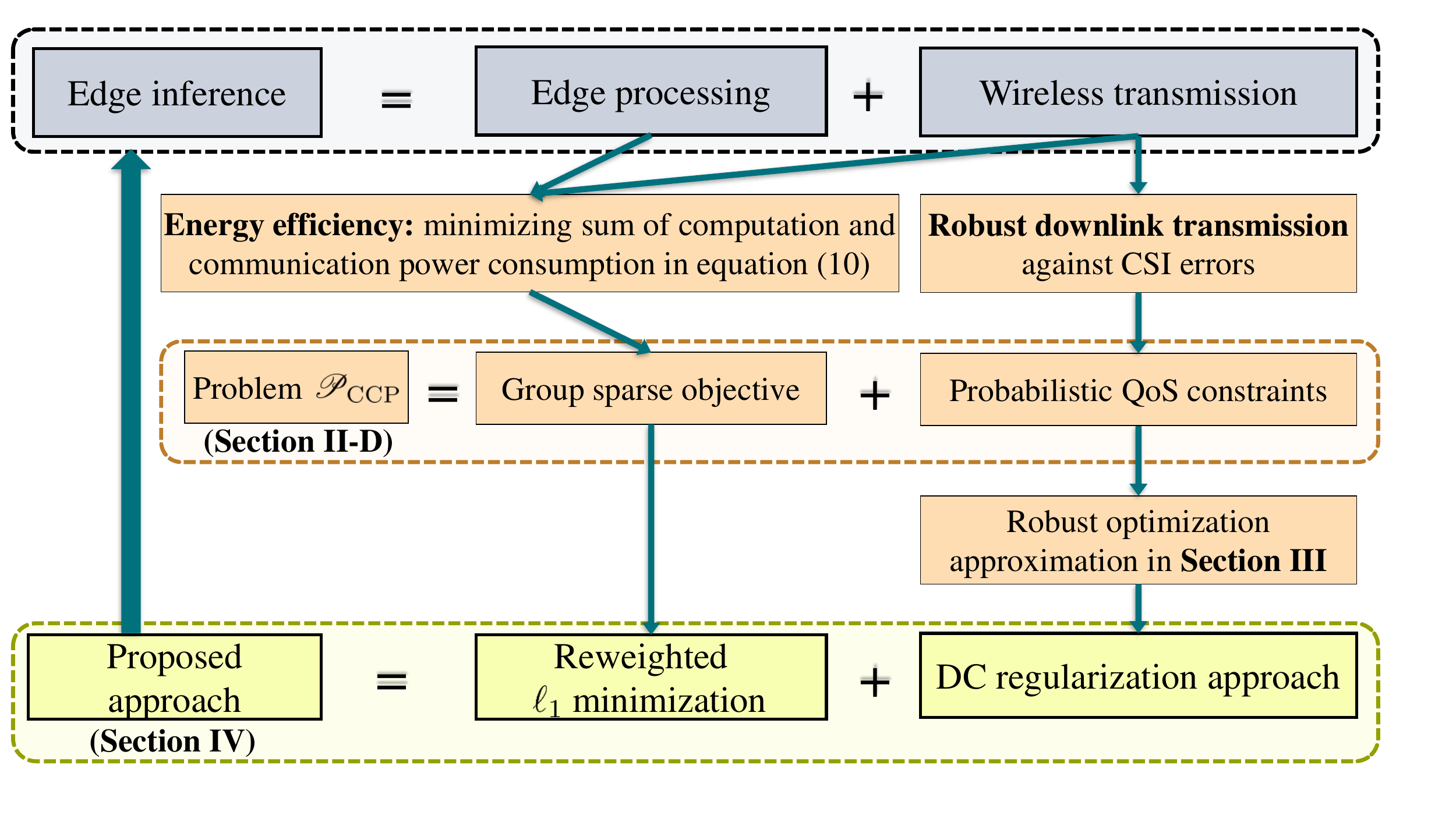}
  \caption{Illustration of our energy-efficient processing and robust wireless cooperative transmission framework for edge inference.}
  \label{fig:framework}
\end{figure}

There is a communication and computation tradeoff for the edge inference system in downlink. In particular, performing an inference task at more edge computing nodes can achieve higher quality-of-service (QoS) through cooperative downlink transmission for delivering the output results to mobile users. This however results in more computation power consumption for executing the deep learning models. We thus propose to jointly decide on the task allocation strategy at edge nodes and design downlink beamforming vectors by minimizing the sum of transmission power consumption and computation power consumption. In particular, the power consumption of deep learning inference tasks can be determined through the estimated energy \cite{yang2017designing} and computation time. We observe that there is an intrinsic connection between the group sparse structure \cite{Yuanming_TWC2014,tao2016content} of the downlink aggregative beamforming vector and the combinatorial variable, i.e., the set of tasks performed at edge nodes. The cooperative transmission strategies require global channel state information (CSI), while uncertainty in CSI acquisition is inevitable in practice due to training based channel estimation \cite{yangfuqian2018pilot}, limited feedback \cite{mo2018limited}, partial CSI acquisition \cite{shi2015optimal} and CSI acquisition delays \cite{maddah2012completely}. We thus formulate the joint task selection and downlink beamforming problem for energy-efficient processing and robust transmission against CSI errors in edge inference system as a group sparse beamforming problem with probabilistic-QoS constraints \cite{shi2015optimal}.

The joint chance constraints make the formulated probabilistic group sparse beamforming problem highly intractable since it has no closed-form expression generally. To address the chance-constrained programs, a number of works focus on finding computationally tractable approximations based on the collected samples of the random variables. A recognized scenario generation (SG) approach \cite{nemirovski2006convex} is proposed that uses a collection of sampled constraints to approximate the original chance constraints. However, SG is over-conservative since the volume of feasible region decreases by increasing the sample size, which leads to the deterioration of its performance. In addition, given the pre-specified probability $1-\epsilon$ and the confidence level $1-\delta$ for the probabilistic-QoS constraints, the required samples size of SG should satisfy $\sum_{i=1}^{NKL-1}\binom{T}{i}\epsilon^i(1-\epsilon)^{T-i}\leq \delta$, which increases roughly linearly with $1/\epsilon$. In \cite{shi2015optimal}, a stochastic optimization approach is provided to address the over-conservativeness of SG. However, its computational cost grows linearly with the sample size, which is not scalable for obtaining high-robustness solutions. Moreover, its statistical guarantee under finite sample size is still not available. 
To overcome limitations of existing methods, we present a robust optimization approximation approach for the joint chance constraints by enforcing the QoS constraints for any element within a high probability region. The high probability region is further determined by adopting a statistical learning \cite{hong2017learning} approach. This approach enjoys the benefits that the minimum required sample size is only $\log\delta/\log(1-\epsilon)$, and the computational cost is independent of the sample size. 

With the statistical learning based robust optimization approximation approach, the resulting robust group sparse beamforming problem has nonconvex quadratic constraints and a nonconvex group sparse objective function. We find that the nonconvex quadratic constraints can be convexified by matrix lifting and semidefinite relaxation (SDR) \cite{luo2007approximation}. Specifically, the nonconvex quadratic robust QoS constraints can be lifted as convex constraints in terms of a rank-one positive semidefinite matrix variable, which is then convexified by simply dropping the rank-one constraint. However, the SDR approach cannot guarantee that the obtained solution is feasible with respect to the original nonconvex quadratic constraints. The mixed $\ell_1/\ell_2$-norm \cite{bach2012optimization} is a well-known convex group sparsity inducing norm, which has been successfully applied in green cloud radio access networks \cite{Yuanming_TWC2014} and cooperative wireless cellular network \cite{sanjabi2014joint}. However, the SDR approach requires a quadratic form of the objective function, which makes the mixed $\ell_1/\ell_2$-norm minimization approach inapplicable. To overcome this problem, a quadratic variational form of weighted mixed $\ell_1/\ell_2$-norm is proposed in \cite{shi2015robust_TSP} to induce group sparsity. Note that \cite{shi2015robust_TSP} also considers a group sparse beamforming problem with nonconvex quadratic constraints. However, the performance of a quadratic variational form of weighted mixed $\ell_1/\ell_2$-norm minimization with SDR is still not satisfactory. 

To address the limitations of existing approaches, we propose a reweighted power minimization approach to enhance the group sparsity as well as improve the feasibility of nonconvex quadratic constraints. Specifically, we first adopt the iteratively reweighted $\ell_1$ minimization approach for enhancing group sparsity \cite{dai2016energy,shi2016smoothed}. To further guarantee the feasibility of the original nonconvex quadratic constraints, we exploit the matrix lifting technique to recast the nonconvex quadratic constraints as the convex constraints with respect to a rank-one positive semidefinite matrix, and propose a novel difference-of-convex-functions (DC) regularzation approach to induce rank-one solutions. Numerical results demonstrate that the proposed approach improves the probability of feasibility by avoiding the over-conservativeness of SG. Benefiting from both the reweighted $\ell_1$ minimization and the DC regularization, the proposed approach achieves a much lower total power consumption than the algorithm proposed in \cite{shi2015robust_TSP} and has a better capability of inducing group sparsity with nonconvex quadratic constraints.

\subsection{Contributions}
In this work, we consider an edge computing system to execute deep learning inference tasks for resource-constrained mobile devices. In order to provide energy-efficient processing and robust wireless cooperative transmission service for edge inference, we propose to jointly design the downlink beamforming vector and the set of inference tasks performed at each edge computing nodes under probabilistic-QoS constraints. We provide a statistical learning based robust optimization approximation for the highly intractable joint chance constraints, which guarantees that the probabilistic-QoS constraints are feasible with certain confidence level. The resulting problem turns out to be a group sparse beamforming problem with nonconvex quadratic constraints. We propose a reweighted power minimization approach based on the principles of iteratively reweighted $\ell_1$ minimization for group sparsity inducing, matrix lifting technique, and a novel DC representation for rank-one positive semidefinite matrices. The proposed approach can enhance group sparsity and induce rank-one solutions.

We summarize the major contributions of this paper as follows:
\begin{enumerate}
  \item We propose an energy-efficient processing and robust transmission approach for executing deep learning inference tasks at possibly multiple edge computing enabled wireless access points. The selection of optimal set of access points for each task is formulated as a group sparse beamforming problem with joint chance constraints.
  \item We provide a robust optimization counterpart to approximate the joint chance constraints followed by a statistical learning approach to learn the parameters from data samples of the random channel coefficients. It turns out a nonconvex group sparse beamforming problem with nonconvex quadratic constraints.
  \item We show that the nonconvex quadratic constraints can be reformulated as convex constraints with a rank-one constraint, where the rank-one constraint can be reformulated with a novel DC representation. 
  To enhance the group sparsity and inducing rank-one solutions, we propose a reweighted power minimization approach by iteratively reweighted $\ell_1$ minimization with DC regularization and updating weights.
  \item We conduct extensive numerical experiments to demonstrate the advantages of the proposed approach in providing energy-efficient and robust transmission service for edge inference.
\end{enumerate}

\subsection{Organization and Notations}
The rest of this work is organized as follows. In Section II, we introduce the system model and the power consumption model of edge inference, and formulate the energy-efficient processing and robust cooperative transmission problem as a group sparse beamforming problem with joint chance constraints. Section II provides a statistical learning based robust optimization approach to approximate the joint chance constraints. In Section IV, we design a reweighted power minimization approach for solving the robust group sparse beamforming problem. The simulation results are illustrated in Section V to demonstrate the superiority of the proposed approach over other state-of-the-art approaches. Finally, we conclude this work in Section VI.

Throughout this paper, we use lower-case bold letters (e.g., $\bs{v}$) to denote column vectors and letters with one subscript to denote their subvectors (e.g., $\bs{v}_{k}$). We further use lower-case bold letters with two subscripts to denote the subvectors of subvectors (e.g., $\bs{v}_{nk}$ is a subvector of $\bs{v}_k$). We denote scalars with lower-case letters,  matrices with capital letters (e.g., $\bs{V}$) and sets with calligraphic letters (e.g., $\mathcal{A}$). The conjugate transpose of a vector or matrix, $\ell_2$-norm of a vector and spectral norm of a matrix are denoted as $(\cdot)^{\sf{H}},\|\cdot\|_2$ and $\|\cdot\|$, respectively. Table \ref{tab:notations} summarizes the notations used in this paper.
\begin{table}[h]
\centering
        \begin{tabular}{|c|p{0.65\columnwidth}|}
        \hline
        Notation  & Explanation \\ \hline
        $N,K,L$ & the number of APs, MUs, and AP's antennas, respectively\\
        $[K]$ & the set of $\{1,\cdots,K\}$ \\
        $\mathcal{A}$ & task allocation of APs\\
        $P_{nk}^c$ & power consumption of performing the $k$-th user's task at the $n$-th AP  \\
        $P_n^{\text{Tx}}$ & maximum transmit power of the $n$-th AP \\
        $\eta_n$ & power amplifier efficiency \\
        $P^c$ & total computation power consumption at APs \\
        $P$ & total power consumption \\
        $\bs{v}_{nk},\bs{v}_{k},\bs{v}$ & beamforming vectors at the APs \\
        $\bs{V}_{ij}[s,t],\bs{V}_{ij},\bs{V}$ & lifted matrices of beamforming vectors \\
        $\bs{h}_{kn},\bs{h}_{k},\bs{h}$ & downlink channel coefficient vectors between APs and MUs \\
        $\hat{\bs{h}}_{kn},\hat{\bs{h}}_k,\hat{\bs{h}}$ & estimated channel coefficient vectors\\
        $\bs{e}_{kn},\bs{e}$ & random errors of CSI \\
        $\gamma_k,\zeta$ & the target QoS and its target tolerance level\\
        $\epsilon,\delta$ & the tolerance level and its confidence level \\
        $\mathcal{U}_k$ & high probability region of $\bs{h}_k$ \\
        $\mathcal{D}$ & the data set consisting of $D$ i.i.d. samples of $\bs{h}$ \\
        $\mathcal{D}^1,\mathcal{D}^2$ & the partitioned two parts of the data set $\mathcal{D}$ with size $D_1$ and $D_2=D-D_1$, respectively \\
        $\tilde{\bs{h}}^{(j)}$ & the $j$-th data sample \\
        $q_{1-\epsilon}$ & $(1-\epsilon)$-quantile
        \\\hline
        \end{tabular}\vspace{0.5em}
\caption{Notations used in the paper}
\label{tab:notations}
\end{table}

\section{System Model and Problem Formulation}
This section provides the system model and power consumption model of edge inference for deep neural networks, followed by the proposal of the energy-efficient edge processing under probabilistic-QoS constraints. 

\subsection{System Model}

Consider the edge processing network consisting of $N$ $L$-antenna edge computing enabled wireless access points (APs) and $K$ single-antenna mobile users (MUs), as shown in Fig. \ref{fig:system}. Each MU $k$ has a deep learning inference task $\phi_k(d_k)$ with input $d_k$. Instead of relying on a cloud data center, we execute deep learning tasks at the APs to address latency and privacy concerns for high-stake applications such as drones and smart vehicles \cite{park2018edgeai}. In this paper, we propose to store the trained deep neural network (DNN) models $\phi_k$'s to APs in advance. Each AP collects all inputs $\{d_k\}_{k=1}^{K}$ from each MU in the first phase. In the second phase,  each AP will selectively execute some inference tasks and transmit the output results to the MUs through cooperative downlink transmission, thereby providing low-latency intelligent services for MUs. The point is that the same inference task can be executed at multiple APs, so that the multiple APs can jointly transmit the result to the MUs through beamforming, thus improving the downlink transmission efficiency (at the expenses of the larger energy consumption due to executing the same task at multiple APs.) This paper focuses on the joint task selection and downlink transmit beamforming problem in the second phase.

\begin{figure}[h]
  \centering
  \includegraphics[width=\columnwidth]{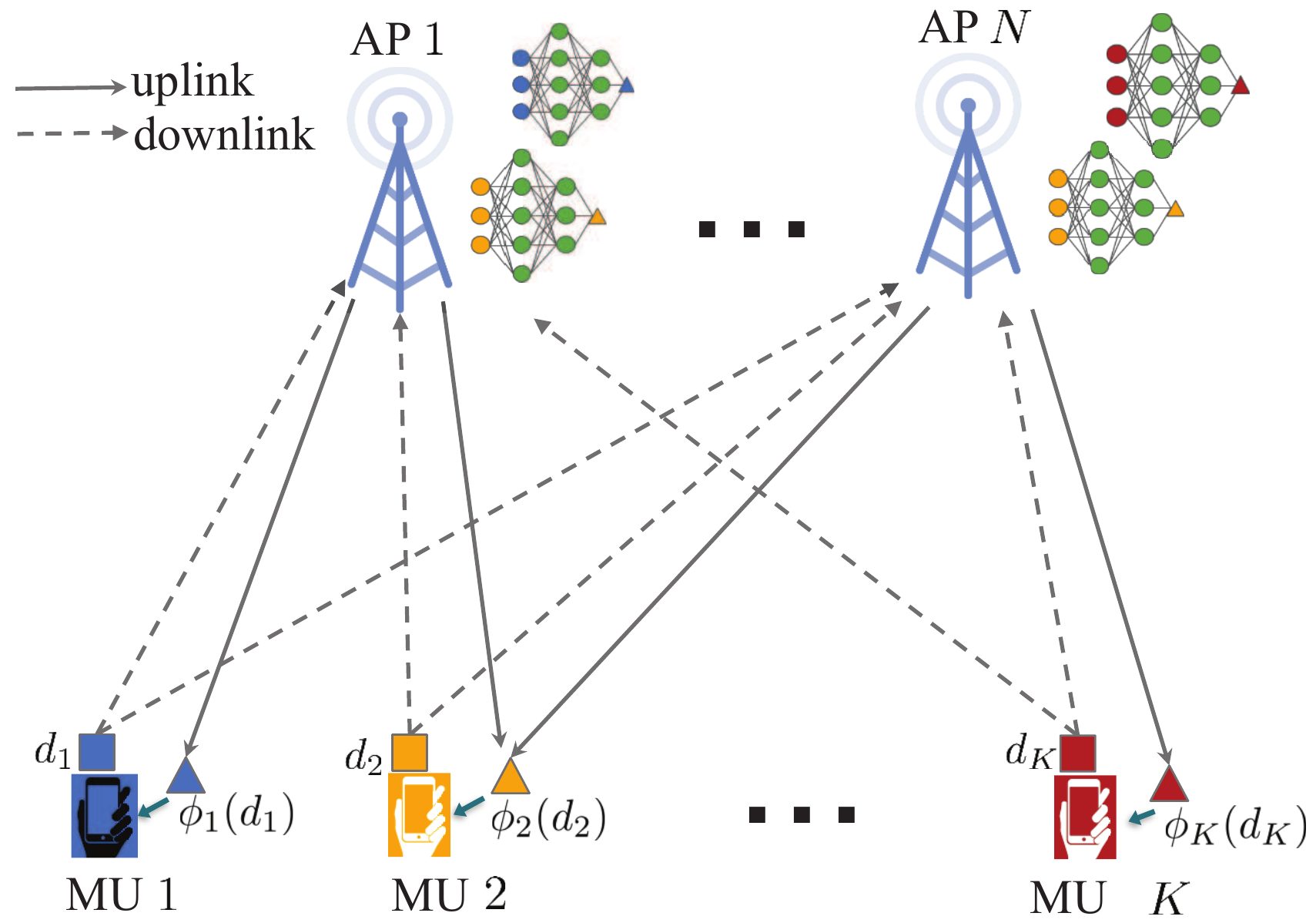}
  \caption{System model of edge inference for deep neural networks. This papper focuses on the computing and downlink transmission phase.}
  \label{fig:system}
\end{figure}

Let $\phi_k(d_k)$ be the requested output for MU $k$, $s_k\in\mathbb{C}$ be the encoded scalar to be transmitted, and $\bs{v}_{nk}\in\mathbb{C}^{L}$ be the beamforming vector for message $\phi_k(d_k)$ at the $n$-th AP. We consider the downlink communication scenario, where all inputs $d_k$'s have already been collected at APs. Then the received signal at MU $l$ is given by
\begin{equation}
    y_k = \sum_{n=1}^{N}\sum_{l=1}^{K}\bs{h}_{kn}^{\sf{H}}\bs{v}_{nl}s_l+z_k, \label{eq:io1}
\end{equation}
where $\bs{h}_{kn}\in\mathbb{C}^{L}$ is the channel coefficient vector between the $n$-th AP and the $k$-th MU, $z_k\sim\mathcal{CN}(0,\sigma_k^2)$ is the additive isotropic white Gaussian noise. Suppose all data symbols $s_k$'s are mutually independent with unit power, i.e., $\mathbb{E}[|s_k|^2]=1$, and also independent with the noise. Denote $[K]$ as the set $\{1,\cdots,K\}$. Let $\mathcal{A}\subseteq\{(n,k):n\in[N],k\in[K]\}$ denote a feasible allocation for the inference tasks on APs, i.e., computational task $\phi_k$ shall be performed at the $n$-th AP for $(n,k)\in\mathcal{A}$. In term of the group sparsity structure of the aggregative beamforming vector 
\begin{equation}
\bs{v}=[\bs{v}_{11}^{\sf{H}},\cdots,\bs{v}_{N1}^{\sf{H}},\cdots,\bs{v}_{NK}^{\sf{H}}]^{\sf{H}}\in\mathbb{C}^{NKL},    
\end{equation}
we have that if the inference task $k$ will not be performed at AP $n$, i.e., $(n,k)\notin\mathcal{A}$, the beamforming vector $\bs{v}_{nk}$ will  be set as zero. Let $\mathcal{T}(\bs{v})$ be the group sparsity pattern of $\bs{v}$ given as 
\begin{equation}
    \mathcal{T}(\bs{v})=\{(n,k)|\bs{v}_{nk}\ne\bs{0}\}.
\end{equation}
The signal-to-interference-plus-noise-ratio (SINR) for mobile device $k$ is given by
 \begin{equation}
     \textrm{SINR}_k(\bs{v};\bs{h}_k) = \frac{|\bs{h}_{k}^{\sf{H}}\bs{v}_{k}|^2}{\sum_{l\ne k}|\bs{h}_{k}^{\sf{H}}\bs{v}_{l}|^2+\sigma_k^2},
 \end{equation}
 where $\bs{h}_{k}$ and $\bs{v}_k$ are given by
 \begin{align}
     \bs{h}_{k}&=[\bs{h}_{k1}^{\sf{H}},\cdots,\bs{h}_{kN}^{\sf{H}}]^{\sf{H}}\in\mathbb{C}^{NL}, \\
     \bs{v}_{k}&=\begin{bmatrix}
    \bs{v}_{1k}^{\sf{H}} & \cdots & \bs{v}_{Nk}^{\sf{H}}
\end{bmatrix}^{\sf{H}}\in\mathbb{C}^{NL},
 \end{align}
and the aggregative channel coefficient vector is denoted as
\begin{equation}
\bs{h}=[\bs{h}_{1}^{\sf{H}},\cdots,\bs{h}_{K}^{\sf{H}}]^{\sf{H}}\in\mathbb{C}^{NKL}.      
\end{equation}  
The transmit power constraint at the $n$-th AP is given by
\begin{equation}
    \mathbb{E}\left[\sum_{l=1}^{K}\|\bs{v}_{nl}s_l\|_2^2\right] = \sum_{l=1}^{K}\|\bs{v}_{nl}\|_2^2\leq P_n^{\text{Tx}}, n\in[N], \label{constraint:transpower}
\end{equation}
where $P_n^{\text{Tx}}$ is the maximum transmit power.

\subsection{Power Consumption Model}
Although widespread applications of deep learning present numerous opportunities for intelligent systems, energy consumption becomes one of the main concerns \cite{xu2018scaling}. Indeed, the energy consumption of performing DNN inference is dominated by the memory access. As pointed out in \cite{han2015deep}, a memory access of 32 bit dynamic random access memory (DRAM) consumes 640pJ, while a cache access of 32 bit static random access memory (SRAM) consumes 5pJ and a 32 bit floating point add operation consumes 0.9pJ. Large DNN models probably cannot fit in the storage of mobile device, which requires more costly DRAM memory accesses. Therefore, small models can be directly deployed on mobile devices but large models are preferably executed at the powerful edge nodes.
Let the power consumption of computing task $\phi_k$ at the $n$-th edge computing node be $P_{nk}^{\text{c}}$. The total computation power consumption for all edge computing nodes is thus given by
\begin{equation}
  P^{\text{c}}=\sum_{n,k}P_{nk}^{\text{c}}I_{(n,k)\in\mathcal{T}(\bs{v})},
\end{equation}
where the indicator function $I$ is $1$ if $(n,k)\in\mathcal{T}(\bs{v})$ and $0$ otherwise. Therefore, the total power consumption consists of transmission power consumption for output results delivery and computation power consumption for deep learning tasks execution, which is given by
\begin{equation}
    P =\sum_{n,k}\frac{1}{\eta_n}\|\bs{v}_{nk}\|_2^2+ \sum_{n,k}P_{nk}^{\text{c}}I_{(n,k)\in\mathcal{T}(\bs{v})},
\end{equation}
where $\eta_n$ is the power amplifier efficiency. 

Deep neural networks especially deep convolutional neural networks (CNNs) becomes an indispensable and the state-of-the-art paradigm for real-world intelligent services. Its high energy cost has attracted much interest in designing energy-efficient structures of neural networks \cite{han2015deep}. Estimating the energy consumption of a neural network is thus critical for inference at the edge, for which an estimation tool is developed in \cite{EnergyEstimationWebsite}. The energy consumption of performing an inference task consists of the computation part and the data movement part \cite{yang2017designing}. The computation energy consumption can be calculated by counting the number of multiply-and-accumulate (MACs) in the layer and weighing it with the energy consumption of each MAC operation in the computation core. The energy consumption of data movement is calculated by counting the number of accessing memory at each level of the memory hierarchy in the corresponding hardware and weighing it with the energy consumption of accessing the memory in the corresponding level. 

Here we illustrate how to estimate the computation power consumption of performing image classification tasks using the classic CNN (i.e., AlexNet consisting of $5$ convolutional layers and $3$ fully-connected layers) on the Eyeriss chip. The energy estimation tool takes network configuration as input and outputs the estimated energy breakdown of each layer in terms of computation part and the data movement part of three data types(weight, input feature map, output feature map). Figure \ref{fig:EnergyDecomp} demonstrates the estimated energy of each layer running on Eyeriss chip, and the overall energy consumption is the sum of four parts. The unit of energy is normalized by the energy for one MAC. Based on the total energy consumption, the computation power consumption can be further determined via dividing the energy consumption by the computation time.

\begin{figure}[!t]
  \centering
  \includegraphics[width=\columnwidth]{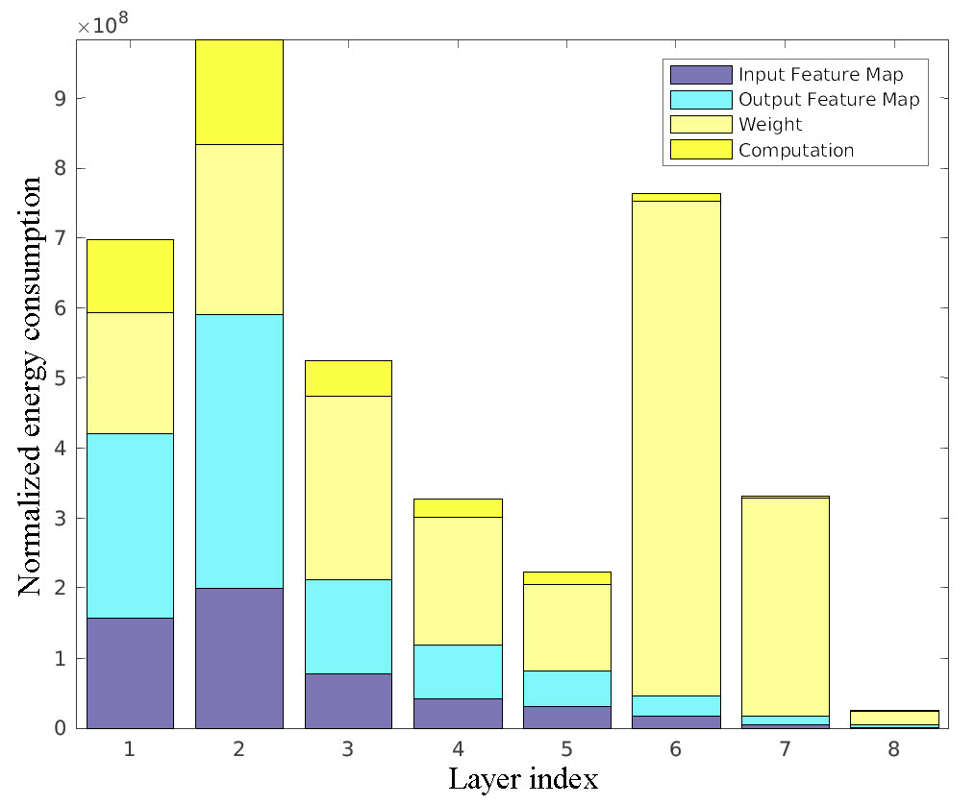}
  \caption{Energy consumption breakdown of the AlexNet \cite{krizhevsky2012imagenet}. The unit of energy is normalized by the energy for one MAC operation (i.e., $10^2$ = energy of 100 MACs).}
  \label{fig:EnergyDecomp}
\end{figure}

\subsection{Channel Uncertainty Model}\label{subsec:channel_uncertainty_model}
For high-stake intelligent applications such as autonomous driving and automation, robustness is a critical requirement. In practice, inevitably there is uncertainty in the available channel state information (CSI) $\bs{h}$, which is taken into consideration to provide robust transmission in this paper. It may originate from training based channel estimation \cite{yangfuqian2018pilot}, limited precision of feedback \cite{mo2018limited}, partial CSI acquisition \cite{shi2015optimal} and delays in CSI acquicition \cite{maddah2012completely}. In this work, we adopt the additive error model \cite{liu2018energy,fang2017joint} of the channel imperfection, i.e.,
\begin{equation}\label{eq:additive_model}
    \bs{h}=\hat{\bs{h}}+\bs{e},
\end{equation}
where $\hat{\bs{h}}\in\mathbb{C}^{NKL}$ is the estimated aggregative channel vector and $\bs{e}\in\mathbb{C}^{NKL}$ is the random errors of the CSI with unknown distribution and expectation as $\bs{0}$. We apply the probabilistic quality-of-service (QoS) constraints \cite{shi2015optimal} to characterize the robustness of delivering the inference results to MUs 
\begin{equation}\label{eq:pcr}
    \textrm{Pr}\left(\textrm{SINR}_{k}(\bs{v};\bs{h}_k) \geq \gamma_k\right)\geq 1-\zeta, \forall k\in[K].
\end{equation}
Here $\zeta$ is the tolerance level and $``\textrm{SINR}_{k} \geq \gamma_k"$ is called safe condition. 

\subsection{Problem Formulation}
In the proposed edge processing framework for deep learning inference tasks, there is a fundamental tradeoff between computation and communication. Specifically, executing the same inference task at multiple edge nodes will require higher computation power consumption, while the downlink transmission power consumption shall be reduced due to the cooperative transmission gains. In this paper, we propose an energy-efficient processing and robust transmission approach to minimize the total network power consumption, while satisfying the probabilistic QoS constraints and transmit power constraints. It is formulated as the following probabilistic group sparse beamforming problem:
\setlength\arraycolsep{1.5pt}
\begin{eqnarray}
    \!\!\!\!\!\!\!\!\!\mathscr{P}_{\text{CCP}}:\!\!\mathop{\textrm{min}}_{\bs{v}\in\mathbb{C}^{NKL}} && \!\!\sum_{n,k}\frac{1}{\eta_n}\|\bs{v}_{nk}\|_2^2+ \sum_{n,k}P_{nk}^{\text{c}}I_{(n,k)\in\mathcal{T}(\bs{v})}\nonumber \\
    \textrm{s.t.~~~~} && \!\!\!\!\textrm{Pr}\left(\textrm{SINR}_{k}(\bs{v};\bs{h}_k) \geq \gamma_k\right)\geq 1-\zeta, k\in[K] \label{constraint:probQoS} \\
    && \sum_{k=1}^{K}\|\bs{v}_{nk}\|_2^2\leq P_n^{\text{Tx}}, n\in[N].
\end{eqnarray}

\begin{remark}
    In edge inference, data privacy is another main concern for high-stake applications such as smart vehicles and drones. Mobile users in these applications may be reluctant to send their raw data to APs. To avoid the exposure of raw data, hierarchical distributed structure has been studied in the literature, such as \cite{li2019edge}, by determining a partition point of a DNN model and deploying the partitioned model across the mobile device and the edge computing enabled AP. The data privacy is protected since only the output of the layers before the partition point is uploaded to APs. Note that our proposed framework is also applicable to the privacy-preserving hierarchical distributed structure. In this case, the input $d_k$ becomes the locally computed output of the layers before the partition point. The computation task $\phi_k$ becomes the task of computing the inference result with the layers after the partition point.
\end{remark}

To achieve the robustness of QoS against CSI errors, we shall collect $D$ i.i.d. (independent and identically distributed) samples of the imperfect channel state information as the data set $\mathcal{D}=\{\tilde{\bs{h}}^{(1)},\cdots,\tilde{\bs{h}}^{(D)}\}$ to learn the uncertainty model of CSI before providing edge inference service. Based on the data set $\mathcal{D}$, we aim to design a beamforming vector $\bs{v}$ such that the safe condition is satisfied with probability at least $1-\zeta$. However, since we do not know the prior distribution of random errors, the statistical guarantee of a given approach is usually expressed as certain confidence level $1-\delta$ for certain tolerance level $1-\epsilon$, e.g., the scenario generation approach \cite{nemirovski2006convex}. That is, the confidence level of 
\begin{equation}\label{relaxed_Prob_QoS}
\textrm{Pr}\left(\textrm{SINR}_{k}(\bs{v};\bs{h}_k) \geq \gamma_k\right)\geq 1-\epsilon    
\end{equation}
is no less than $1-\delta$ for some $\bs{v}$, $D$, $0<\epsilon<1$ and $0<\delta<1$. Thus the violation probability of the safe condition is upper bounded by
\begin{equation}
    \textrm{Pr}(\textrm{SINR}_{k}(\bs{v};\bs{h}_k) < \gamma_k)< \delta + \epsilon(1-\delta). 
\end{equation}
By setting $\epsilon$ and $\delta$ such that $\zeta>\delta+\epsilon(1-\delta)$, the safe condition (\ref{eq:pcr}) is guaranteed to be met.

We consider the block fading channel where the channel distribution is  assumed invariant \cite{liu2019two} within $T_s$ blocks and the channel coefficient vector remains unchanged within each block. Note that the training by collecting $D$ channel samples within each block will result in high signaling overhead. We will show that our proposed approach for addressing the probabilistic-QoS constraints can be intergrated with a cost-effective channel sampling strategy in Section \ref{subsec:sampling_strategy}.
\subsection{Problem Analysis}\label{subsec:problem_analysis}
Directly solving the joint chance constraints (\ref{constraint:probQoS}) is usually a highly-intractable task \cite{nemirovski2006convex}, especially when there is no exact knowledge about the uncertainty. In this work, we shall propose a general framework for edge inference without assuming the prior distribution of random errors. A natural idea is to find a computationally tractable approximation for the probabilistic QoS constraints (\ref{constraint:probQoS}).

\subsubsection{Scenario Generation}
Scenario generation \cite{nemirovski2006convex} is a well-known approach by obtaining $D$ independent samples of the random channel coefficient vector $\bs{h}$ and imposing the target QoS constraints $\textrm{SINR}_{k} \geq \gamma_k,k\in[K]$ for each sample. However, because it ensures robustness in the minimax sense, it is too conservative when a large number of samples are drawn, since the volumn of feasible region will decrease, which may result in the infeasibility of problem $\mathscr{P}_{\text{CCP}}$. In addition, the sample size $D$ should be chosen such that $\sum_{i=1}^{NKL-1}\binom{D}{i}\epsilon^i(1-\epsilon)^{D-i}\leq \delta$, where $1-\delta$ gives the confidence level for the probabilistic-QoS constraints defined in equation (\ref{eq:pcr}). Therefore, the scenario generation approach has scalability issue since the required minimum sample size $D$ increases roughly linearly with $1/\epsilon$ for small $\epsilon$ and also with $NKL$. 

\subsubsection{Stochastic Programming}
To address this over-conservativeness issue of the scenario generation approach, a stochastic programming approach is further provided in \cite{shi2015optimal} by finding a difference-of-convex-functions (DC) approximation for the chance constraints. The resulting DC constrained stochastic program can be solved by successive convex approximation with the Monte Carlo approach at each iteration. However, its computation cost grows linearly with the number of samples $D$ which is not scalable for obtaining high-robustness solutions, and the statistical guarantee is not available for the joint chance constraints under finite sample size.

To address the limitations of the existing works, we shall present a robust optimization approach in Section \ref{sec:ropt} to approximate the chance constraint via a statistical learning approach\cite{hong2017learning}. This approach enjoys the main advantages that the minimum required number of observations is only $\log\delta/\log(1-\epsilon)$ and the computational cost is independent of the sample size.

\section{Learning-Based Robust Optimization Approximation for Joint Chance Constraints}\label{sec:ropt}
In this section, we provide a robust optimization approximation for the joint chance constraints in problem $\mathscr{P}_{\text{CCP}}$, followed by a statistical learning approach to learn the shape and size of the high probability region. 


\subsection{Approximating Joint Chance Constraints via Robust Optimization}
Robust optimization \cite{hong2017learning} uses safe approximation and imposes that the safe conditions are always satisfied when the random variables lie in some geometric set. Specifically, the robust optimization approximation of the joint chance constraints (\ref{constraint:probQoS}) is given by
\begin{equation}
\textrm{SINR}_{k}(\bs{v};\bs{h}_k) \geq \gamma_k,  \bs{h}_k \in \mathcal{U}_k,\forall k \in [K]\label{constraint:RO_QoS}
\end{equation}
where $\mathcal{U}_k$ is the high probability region that $\bs{h}_k$ lies in. 
The robust optimization approximation for the joint chance constraints should yield a solution such that the probabilistic QoS constraint is satisfied with high confidence. The robust optimization approximation approach is realized by constructing a high probability region $\mathcal{U}_k$ from the data set $\mathcal{D}$ such that $\mathcal{U}_k$ covers a $1-\epsilon$ content of $\bs{h}_k$, i.e., 
\begin{equation}\label{eq:prob_cover}
\textrm{Pr}(\bs{h}_k \in \mathcal{U}_k)\geq 1-\epsilon,
\end{equation}
with confidence level at least $1-\delta$. More precisely, since $\mathcal{U}_k$ is generated from data and therefore is random, we require that the proportion of time (\ref{eq:prob_cover}) is satisfied to be at least $1-\delta$ in the repeated application of the data generation and high probability region construction procedure. By imposing the QoS constraints for element in the high probability region as presented in equation (\ref{constraint:RO_QoS}), the confidence level for the probabilistic-QoS constraints (\ref{relaxed_Prob_QoS}) will be at least $1-\delta$.
We thus obtain the robust optimization approximation for problem $\mathscr{P}_{\text{CCP}}$ as
\setlength\arraycolsep{2pt}
\begin{eqnarray}
    \!\!\!\mathscr{P}_{\text{RO}}:\mathop{\textrm{minimize}}_{\bs{v},\bs{h}} && \sum_{n,k}\frac{1}{\eta_n}\|\bs{v}_{nk}\|_2^2+ \sum_{n,l}P_{nk}^{\text{c}}I_{(n,k)\in\mathcal{T}(\bs{v})}\nonumber \\
    \textrm{subject to} && \textrm{SINR}_{k}(\bs{v};\bs{h}_k) \geq \gamma_k, \bs{h}_k \in \mathcal{U}_k, k\in[K] \nonumber \\
    && \sum_{k=1}^{K}\|\bs{v}_{nk}\|_2^2\leq P_n^{\text{Tx}}, n\in[N].
\end{eqnarray}

The choice of the geometric shape of the uncertainty set $\mathcal{U}_k$ is critical to the performance and the tracatability of the robust optimization approximation. Motivated by the tractability of robust optimization, ellipsoids and polytopes are commonly chosen as the basic uncertainty sets. The uncertainty set can be further augmented as the unions or intersection of these basic sets. In this paper, we choose ellipsoidal uncertainty set to model the uncertainty of each group of channel coefficient vector $\bs{h}_{k}$ for its wide use in modeling CSI uncertainties \cite{shi2015robust_TSP,hanif2013efficient}, as well as its tractability as shown in Section \ref{subsec:tractable_reformulation}. The high probability region $\mathcal{U}_k$ is parameterized as
\begin{equation}
    \mathcal{U}_{k}=\{\bs{h}_{k}:\bs{h}_{k}=\hat{\bs{h}}_{k}+\bs{B}_{k}\bs{u}_{k}, \bs{u}_{k}^{\sf{H}}\bs{u}_{k} \leq 1\}. \label{eq:uncertainty_set1}
\end{equation}
Here the parameters $\bs{B}_{k}\in\mathbb{C}^{NL\times NL}$ and $\hat{\bs{h}}_{k}\in\mathbb{C}^{NL}$ shall be learned from the data set $\mathcal{D}$, which will be presented in Section \ref{subsec:learninguncertainty}. We will then present the tractable reformulation of the robust optimization counterpart problem $\mathscr{P}_{\text{RO}}$ in Section \ref{subsec:tractable_reformulation}.

\subsection{Learning the High Probability Region from Data Samples}\label{subsec:learninguncertainty}
Note that (\ref{constraint:RO_QoS}) only gives a feasibility guarantee for the joint chance constraints with statistical confidence at least $1-\delta$, but its conservativeness is still a challenging problem. Generally speaking, problem $\mathscr{P}_{\text{RO}}$ is a less conservative approximation for problem $\mathscr{P}_{\text{CCP}}$ if it has a larger feasible region. Therefore, we prefer a smaller volume of the high probability region $\mathcal{U}$ which provides a larger feasible region. In our problem formulation, we set the volume of the high probability region such that the statistical confidence for the probabilistic-QoS constraints is close to $1-\delta$.

In this paper, we propose to use a statistical learning approach \cite{hong2017learning} for the parameters of the high probability region $\mathcal{U}$, which consists of a shape learning procedure and a size calibration procedure via quantile estimation. First of all, we split the samples in data set $\mathcal{D}$ into two parts, i.e., $
\mathcal{D}^{1}=\{\tilde{\bs{h}}^{(1)},\cdots,\tilde{\bs{h}}^{(D_1)}\}$ and $
\mathcal{D}^{2}=\{\tilde{\bs{h}}^{(D_1+1)},\cdots,\tilde{\bs{h}}^{(D)}\}$, each for one procedure.

\subsubsection{Shape Learning}
Each ellipsoid set $\mathcal{U}_{k}$ can be re-parameterized as
\begin{equation}
    \mathcal{U}_{k}=\{\bs{h}_{k}:(\bs{h}_{k}-\hat{\bs{h}}_{k})^{T}\bs{\Sigma}_{k}^{-1}(\bs{h}_{k}-\hat{\bs{h}}_{k})\leq s_{k}\},
\end{equation}
where $\hat{\bs{h}}_{k}$ and $\bs{\Sigma}_{k}$ are shape parameters of the ellipsoid $\mathcal{U}_{k}$, $s_{k}>0$ determines its size, and $\bs{\Sigma}_{k}/{s}_{k}=\bs{B}_{k}\bs{B}_{k}^{\sf{H}}$. Suppose the observations of $\bs{h}_{k}$ is given by $\mathcal{D}_{k}=\mathcal{D}_{k}^{1}\cup \mathcal{D}_{k}^{2}=\{\tilde{\bs{h}}_{k}^{(j)}\}_{j=1}^{D}$. The shape parameter $\hat{\bs{h}}_{k}$ can be chosen as the sample mean, i.e.,
\begin{align}
  \hat{\bs{h}}_{k}&=\frac{1}{D_1}\sum_{j=1}^{D_1}\tilde{\bs{h}}_{k}^{(j)},\label{eq:HPR_mean}
\end{align}
To reduce the complexity of the ellipsoid, we omit the correlation between each $\{\bs{h}_{kn}\}$ and choose $\bs{\Sigma}_{k}$ as the block diagonal matrix where each diagonal element is the sample covariance of the first part of the data set for $\bs{h}_{kn}$, i.e., 
\begin{align}
     \bs{\Sigma}_{k}&=\begin{bmatrix}
         \bs{\Sigma}_{k1} & &  \\
         & \ddots &  \\
         & & \bs{\Sigma}_{kN}
     \end{bmatrix}, ~\text{where}~~\nonumber\\
     \bs{\Sigma}_{kn}&=\frac{1}{D_1-1}\sum_{j=1}^{D_1}(\tilde{\bs{h}}_{kn}^{(j)}-\hat{\bs{h}}_{kn})(\tilde{\bs{h}}_{kn}^{(j)}-\hat{\bs{hn}}_{k})^{\sf{H}}.\label{eq:HPR_covariance}
\end{align}

\subsubsection{Size Calibration via Quantile Estimation}
We then use the second part of data set $\mathcal{D}_{k}^2$ for calibrating the ellipsoid size $s_{k}$. The key idea is to estimate a $1-\epsilon$ quantile with $1-\delta$ confidence of a transformation of the data samples in $\mathcal{D}_{k}^2$. Let 
\begin{equation}\label{eq:map}
  \mathcal{G}(\xi)=(\xi-\hat{\bs{h}}_{k})^{T}\bs{\Sigma}_{k}^{-1}(\xi-\hat{\bs{h}}_{k})
\end{equation}
be the map from the random space that $\bs{h}_{k}$ lies in to $\mathbb{R}$. The size parameter $s_{k}$ will be chosen as an estimated $(1-\epsilon)$-quantile of the underlying distribution of $\mathcal{G}(\xi)$ based on the data samples in $\mathcal{D}_{nk}^2$, where the $(1-\epsilon)$-quantile $q_{1-\epsilon}$ is defined from
\begin{equation}
    \textrm{Pr}(\mathcal{G}(\xi)\leq q_{1-\epsilon}) = 1-\epsilon.
\end{equation}
Specifically, by computing the function values of $\mathcal{G}$ on each sample of $\mathcal{D}_{k}^2$, we can obtain the observations $G_{1},\cdots,G_{D-D_1}$ where $G_{j}=\mathcal{G}(\bs{h}_{k}^{(D_1+j)})$. Then the $t^\star$-th value of the ranked observations $G_{(1)}\leq \cdots \leq G_{(D-D_1)}$ in ascending order, denoted as $G_{(j^\star)}$, can be an upper bound of the $(1-\epsilon)$-quantile of the underlying distribution of $G(\xi)$ based on the following proposition:
\begin{proposition}
  $s_{k}$ is an upper bound of the $(1-\epsilon)$-quantile of the underlying distribution with $1-\delta$ confidence, i.e.,
\begin{equation}
\textrm{Pr}(s_{k}\geq q_{1-\epsilon})\geq 1-\delta,
\end{equation}
if $s_{k}$ is set as
\begin{align}
&s_{k} = G_{(j^\star)},~~\text{where}~j^\star~\text{is given by} \nonumber\\
&\min_{1\leq j\leq D-D_1}\!\!\left\{j:\sum_{k=0}^{j-1}\binom{D-D_1}{k}(1-\epsilon)^k\epsilon^{D-D_1-k}\geq 1-\delta  \right\}.\label{eq:HPR_size}
\end{align}
\end{proposition}
\begin{proof}
    According to the definition of the quantile $q_{1-\epsilon}$, we have
    \begin{align}
        &\textrm{Pr}(G_{(j)}\geq q_{1-\epsilon}) \nonumber\\
        =&\textrm{Pr}(G_{(k)}<q_{1-\epsilon}, k=0,\cdots, j-1)\nonumber\\
        =&\sum_{k=0}^{j-1}\binom{D-D_1}{k}(1-\epsilon)^k\epsilon^{D-D_1-k}.
    \end{align}
    Therefore $G_{(j^\star)}$ is the smallest one among all upper bounds of the $(1-\epsilon)$-quantile of the underlying distribution with $1-\delta$ confidence. 
\end{proof}

Using the presented two procedures, we learn a high probability region $\mathcal{U}$ of the random channel coefficient vector $h_{k}$'s. The statistical guarantee of this statistical learning based robust optimization approximation approach is given by the following proposition:
\begin{proposition}
    Suppose the data samples in the data set $D_{k}$ are i.i.d. and chosen from a continuous distribution for any $k$. The data set is split into two independent parts $\mathcal{D}_{k}^1$ and $\mathcal{D}_{k}^2$. Each uncertainty set is chosen as $\mathcal{U}_{k}=\{\bs{h}_{k}:(\bs{h}_{k}-\hat{\bs{h}}_{k})^{T}\bs{\Sigma}_{k}^{-1}(\bs{h}_{k}-\hat{\bs{h}}_{k})\leq s_{k}\}$. Their parameters $\hat{\bs{h}}_{k},\bs{\Sigma}_{k},$ and $s_{k}$ are determined following equation (\ref{eq:HPR_mean}), equation (\ref{eq:HPR_covariance}), and equation (\ref{eq:HPR_size}), respectively. Thus, any feasible solution to problem $\mathscr{P}_{RO}$ guarantees that the probabilistic-QoS constraints (\ref{relaxed_Prob_QoS}) are satisfied with confidence at least $1-\delta$.
\end{proposition}
\begin{proof}
  Since $\mathcal{G}$ depends only on $\mathcal{D}_{k}^{1}$, we have 
  \begin{align}
    &\textrm{Pr}_{\mathcal{D}_{k}^{2}}(\bs{v}\in\mathcal{V}) = \textrm{Pr}_{\mathcal{D}_{k}^{2}}(G_{(t^\star)}\geq q_{1-\epsilon})  \geq 1-\delta.
  \end{align}
  Therefore, it is readily obtained that $\textrm{Pr}(\textrm{SINR}_k\geq \gamma_k)\geq 1-\epsilon$ satisfies with confidence at least $1-\delta$.
\end{proof}

Note that $j^\star$ exists only if 
\begin{equation}
\sum_{k=0}^{D-D_1-1}\binom{D-D_1}{k}(1-\epsilon)^k\epsilon^{D-D_1-k}\geq 1-\delta,  
\end{equation}
which implies that $1-(1-\epsilon)^{D-D_1}\geq 1-\delta$. In other words, the required minimum number of samples is $D>D-D_1\geq \log{\delta}/\log{(1-\epsilon)}$ to achieve the $1-\delta$ confidence of the probabilistic QoS constraint (\ref{constraint:probQoS}). Matrix $\bs{B}_{k}$ can be computed as 
\begin{equation}\label{eq:HPR_B}
    \bs{B}_{k}=\sqrt{s_{k}}\bs{\Delta}_{k},
\end{equation} 
where $\bs{\Delta}_k$ is the Cholesky decomposition of $\bs{\Sigma}_{k}$, i.e., $\bs{\Sigma}_{k}=\bs{\Delta}_{k}\bs{\Delta}_{k}^{\sf{H}}$. We summarize the whole procedure for learning the high probability region $\mathcal{U}$ from data set $\mathcal{D}$ in Algorithm \ref{algorithm:robust}.

\SetNlSty{textbf}{}{:}
\IncMargin{1em}
\begin{algorithm}[h]
\textbf{Input:} the data set $\mathcal{D}=\{\tilde{\bs{h}}^{(1)},\cdots,\tilde{\bs{h}}^{(D)}\}$.\\
\For{each $k=1,\cdots K$}{
\textbf{Data splitting:} Randomly split the samples of $\bs{h}_{k}$, namely $\mathcal{D}_{k}$, into two parts $\mathcal{D}_{k}^1$ and $\mathcal{D}_{k}^2$.\\[2mm]
\textbf{Shape learning:} Set the shape parameters $ \hat{\bs{h}}_{k}$ and $\bs{\Sigma}_{k}$ as equation (\ref{eq:HPR_mean}) and equation (\ref{eq:HPR_covariance}) based on $\mathcal{D}_{k}^1$. \\[2mm]
\textbf{Size calibration:} Set the size parameter $s_{k}$ as $G_{(j^\star)}$ by computing the values of function $\mathcal{G}$ on $\mathcal{D}_{k}^2$, where $\mathcal{G}$ is defined in equation (\ref{eq:map}) and $j^\star$ is chosen as equation (\ref{eq:HPR_size}).\\[2mm] 
Compute $\bs{B}_{k}=\sqrt{s_{k}}\bs{\Delta}_{k}$ through Cholesky decomposition $\bs{\Sigma}_{k}=\bs{\Delta}_{k}\bs{\Delta}_{k}^{\sf{H}}$.}
 \textbf{Output:} $\hat{\bs{h}}_{k}$, $\bs{B}_{k}$ for all $k$.
 \caption{Statistical Learning Based Approach for the High Probability Region $\mathcal{U}_k$}
 \label{algorithm:robust}
\end{algorithm}

\subsection{Tractable Reformulations for Robust Optimization Problem}\label{subsec:tractable_reformulation}
According to the ellipsoidal uncertainty model (\ref{eq:uncertainty_set1}), the robust optimization approximation (\ref{constraint:RO_QoS}) can be rewritten as
\begin{align}
&\bs{h}_k^{\sf{H}}\bigg(\frac{1}{\gamma_k}\bs{v}_k\bs{v}_k^{\sf{H}}-\sum_{l\ne k} \bs{v}_l\bs{v}_l^{\sf{H}}\bigg)\bs{h}_k\geq \sigma_k^2  \label{eq:QoS1} \\
    &\bs{h}_{k} = \hat{\bs{h}}_{k}+ \bs{B}_{k}\bs{u}_{k}, \bs{u}_{k}^{\sf{H}}\bs{u}_{k} \leq 1, \label{eq:uncertainty1}
\end{align}
where $\bs{u}_{nk}\in\mathbb{C}^{L}$. By defining matrices 
\begin{equation}\label{eq:HPR_H}
    \bs{H}_{k}=\begin{bmatrix}
    \hat{\bs{h}}_{k} & \bs{B}_{k}
  \end{bmatrix}\in\mathbb{C}^{NL\times (NL+1)}
\end{equation}
and using the S-procedure \cite{boyd2004convex}, we obtain the following equivalent tractable reformulation for (\ref{eq:QoS1}) and (\ref{eq:uncertainty1}):
\begin{align}
  &\bs{H}_{k}^{\sf{H}}\bigg(\frac{1}{\gamma_k}\bs{v}_{k}\bs{v}_{k}^{\sf{H}}-\sum_{l\ne k}\bs{v}_{l}\bs{v}_{l}^{\sf{H}}\bigg)\bs{H}_{k}\succeq \bs{Q}_k \label{constraint:nonconvex_sdp} \\
  &\lambda_k\geq 0, \label{constraint:non_negative}
\end{align}
where $\bs{\lambda}=[\lambda_{k}]\in\mathbb{R}_{+}^{K}$ and $\bs{Q}_k$ is given by
\begin{equation}
  \bs{Q}_k=\begin{bmatrix}
    \lambda_{k}+\sigma_k^2 & \bs{0} \\
    \bs{0} & -\lambda_{k}\bs{I}_{NL}
  \end{bmatrix}\in\mathbb{C}^{(NL+1)\times(NL+1)}.
\end{equation}
The derivation details of (\ref{constraint:nonconvex_sdp}) and (\ref{constraint:non_negative}) from (\ref{eq:QoS1}) and (\ref{eq:uncertainty1}) is relegated to Appendix \ref{append:Sprocedure}.

Thus the proposed robust optimization approximation for problem $\mathscr{P}_{\text{CCP}}$ is given by the following group sparse beamforming problem with nonconvex quadratic constraints:
\begin{eqnarray}
    \mathscr{P}_{\text{RGS}}:\!\!\!\!\!\mathop{\textrm{minimize}}_{\bs{v}\in\mathbb{C}^{NKL},\bs{\lambda}\in\mathbb{R}^{K}}\!\!\! && \sum_{n,l}\frac{1}{\eta_n}\|\bs{v}_{nl}\|_2^2+ \sum_{n,l}P_{nl}^{\text{c}}I_{(n,l)\in\mathcal{T}(\bs{v})} \nonumber \\
    \textrm{subject to~~~} && (\ref{constraint:nonconvex_sdp}),\lambda_k\geq0, \forall k\in[K]   \label{con:nonconvex_qudratic} \\
    && \sum_{l=1}^{K}\|\bs{v}_{nl}\|_2^2\leq P_n^{\text{Tx}}, \forall n\in[N]. \label{con:transmit_power}
\end{eqnarray}
Its computational complexity of solving problem $\mathscr{P}_{\text{RGS}}$ is independent of the sample size $D$. An effective approach for obtaining approximate solution of nonconvex quadratic constrained quadratic program is to lift the aggregative beamforming vector as a rank-one positive semidefinite matrix $\bs{V}=\bs{v}\bs{v}^{\sf{H}}$ and simply drop the rank-one constraint, which is termed as the semidefinite relaxation (SDR) technique \cite{luo2007approximation}. The obtained solution however may be infeasible for the original nonconvex quadratic constraints. To induce the group sparsity with nonconvex quadratic constraints, a quadratic variational form of weighted mixed $\ell_1/\ell_2$-norm is adopted in \cite{shi2015robust_TSP}. In this paper, we will adopt an iteratively reweighted minimization approach which has demonstrated its effectiveness in cloud radio access network \cite{dai2016energy,shi2016smoothed} to further enhance the group sparsity of the aggregative beamforming vector. In addition, to improve the feasibility for the nonconvex quadratic constraint for each subproblem of the reweighted approach, we shall provide a novel difference-of-convex-functions (DC) approach for inducing rank-one solution. It should be mentioned that the uplink-downlink duality is not applicable to efficiently address the robust QoS constraints (\ref{constraint:nonconvex_sdp}) due to the CSI uncertainty.

\subsection{Integrating the Robust Optimization Approximation with a Cost-Effective Sampling Strategy}\label{subsec:sampling_strategy}

Consider the block fading channel where the channel distribution is assumed invariant \cite{liu2019two} within the \textit{coherence interval for channel statistics}. The coherence interval for channel statistics consists of $T_s$ blocks, where each block is called a \textit{coherence interval for CSI} and the channel coefficient vector remains unchanged within each block. However, collecting $D$ channel samples within each block leads to high signaling overhead. To address this issue, we provide a cost-effective sampling strategy for enabling robust transmission, whose timeline is illustrated in Fig. \ref{fig:time_scale}. 

At the beginning of the coherence interval for channel statistics, we collect $D$ i.i.d. channel samples as $\mathcal{D}$. Based on the data set $\mathcal{D}$, we can learn the estimated channel coefficient vector $\hat{\bs{h}}_k$ from equation (\ref{eq:HPR_mean}) and the estimated high probability region of the error $\bs{e}_k$ as $\bs{B}_{k}$ from equation (\ref{eq:HPR_B}). For the transmission in the first block, we can obtain $\{\bs{H}_k\}$ by combining these two parts following equation (\ref{eq:HPR_H}) and solve the resulting problem $\mathscr{P}_{\text{RGS}}$. For any other block $t>1$, we can obtain the estimated channel coefficient $\hat{\bs{h}}[t]$ as the sample mean by collecting as few as one sample of the channel coefficient vector. By replacing the estimated channel coefficient $\hat{\bs{h}}$ and keeping the error information $\{\bs{B}_{k}:k\in[K]\}$, we can construct the parameter $\{\bs{H}_k[t]\}$ at the $t$-th block as
\begin{equation}\label{eq:efficient_sampling}
    \bs{H}_{k}[t]=\begin{bmatrix}
    \hat{\bs{h}}_{k}[t] & \bs{B}_{k}
  \end{bmatrix},\forall k\in[K],
\end{equation}
and design the transmitter beamformer by solving problem $\mathscr{P}_{\text{RGS}}(\{\bs{H}_{k}[t]\})$, which significantly reduces the signaling overhead for channel sampling. The effectiveness of this cost-effective scheme will be demonstrated in Section \ref{subsec:sim_uncertainty} numerically.

\begin{figure}[h]
  \centering
  \includegraphics[width=\columnwidth]{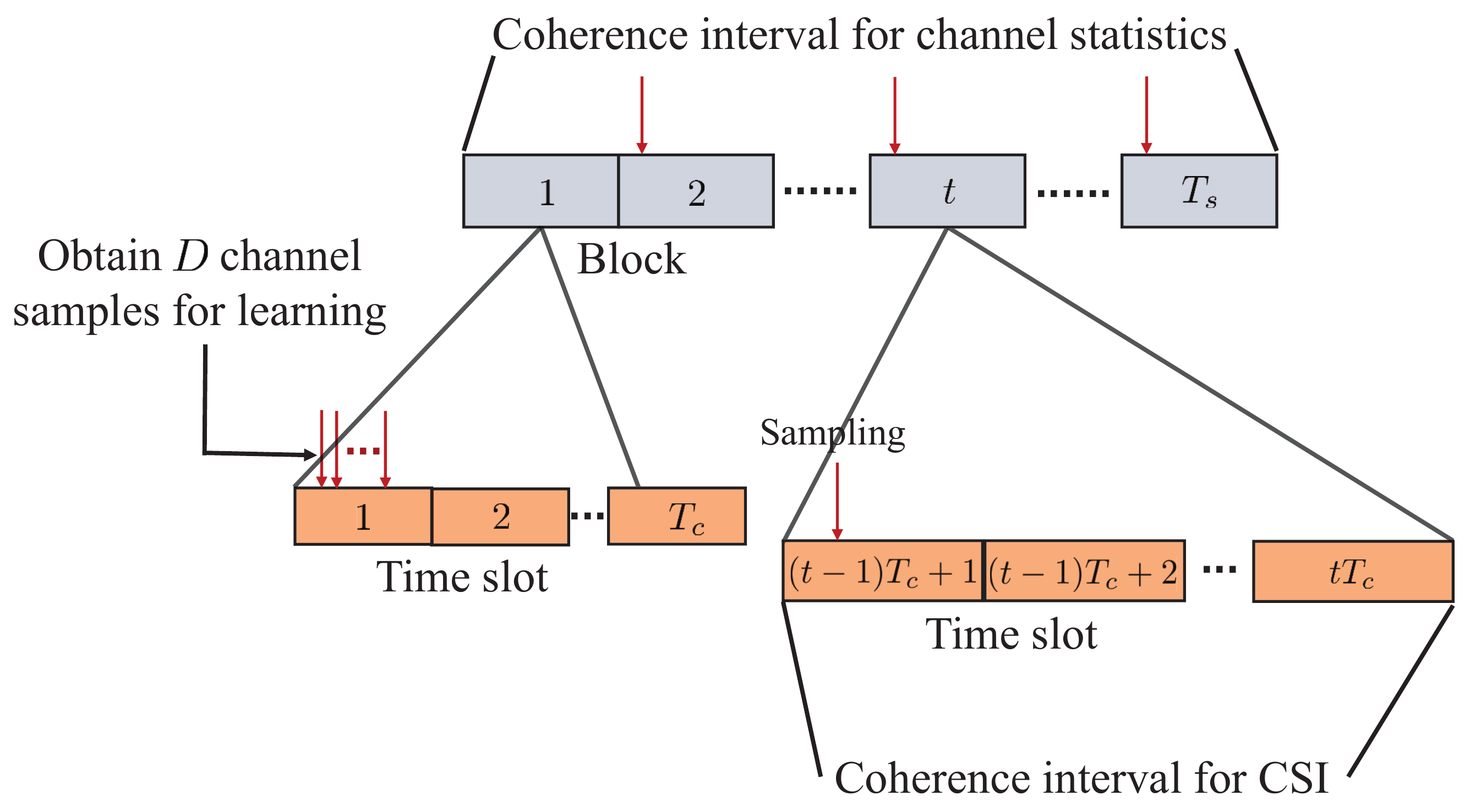}
  \caption{Timeline of a cost-effective channel sampling strategy.}
  \label{fig:time_scale}
\end{figure}

\section{Reweighted Power Minimization for Group Sparse Beamforming with Nonconvex Quadratic Constraints}
This section presents a reweighted power minimization approach to induce the group sparsity structure for problem $\mathscr{P}_{\text{RGS}}$. We further demonstrate that the nonconvex quadratic constraints can be reformulated as convex constraints with respect to a rank-one positive semidefinite matrix using a matrix lifting technique, followed by proposing a DC approach to induce rank-one solutions.

\subsection{Matrix Lifting for Nonconvex Quadratic Constraints}
We observe that constraints (\ref{constraint:nonconvex_sdp}) are convex with respect to $\bs{v}\bs{v}^{\sf{H}}$ despite of its nonconvexity with respect to $\bs{v}$. This motivates us to adopt the matrix lifting technique \cite{luo2007approximation} to address the nonconvex quadratic constraints in problem $\mathscr{P}_{\text{RGS}}$ by denoting 
\begin{align}
  &\bs{V}_{ij}[s,t]=\bs{v}_{si}\bs{v}_{tj}^{\sf{H}}\in\mathbb{C}^{L\times L}\\
  &\bs{V}_{ij}=\begin{bmatrix}
    \bs{V}_{ij}[1,1] & \cdots & \bs{V}_{ij}[1,N] \\
    \vdots & \ddots & \vdots \\
    \bs{V}_{ij}[N,1] & \cdots & \bs{V}_{ij}[N,N]
  \end{bmatrix}=\bs{v}_{i}\bs{v}_{j}^{\sf{H}}\in\mathbb{C}^{NL\times NL}\\
  &\bs{V}=\bs{v}\bs{v}^{\sf{H}}=\begin{bmatrix}
    \bs{V}_{11} & \cdots & \bs{V}_{1K} \\
    \vdots & \ddots & \vdots \\
    \bs{V}_{K1} & \cdots & \bs{V}_{KK}
  \end{bmatrix}\in\mathbb{S}_+^{NKL},
\end{align}
where $\mathbb{S}_+^{NKL}$ denotes the set of Hermitian positive semidefinite (PSD) matrices. The aggregative beamforming vector $\bs{v}$ is thus lifted as a rank-one PSD matrix $\bs{V}$. The constraint $\mathcal{C}_k$ of problem $\mathscr{P}_{\text{RGS}}$, which given by (\ref{constraint:nonconvex_sdp}), can be equivalently rewritten as the following PSD constraint  
\begin{equation}\label{constraint:lifted_sdp}
  \bs{H}_{k}^{\sf{H}}\bigg(\frac{1}{\gamma_k}\bs{V}_{kk}-\sum_{l\ne k}\bs{V}_{ll}\bigg)\bs{H}_{k}\succeq \bs{Q}_k, 
\end{equation}
and the transmit power constraint (\ref{con:transmit_power}) can be equivalently rewritten as
\begin{equation}
  \sum_{l=1}^{K}\|\bs{v}_{nl}\|_2^2=\sum_{l=1}^{K}\textrm{Tr}(\bs{V}_{ll}[n,n])\leq P_n^{\text{Tx}}, \forall n=1,\cdots,N. 
\end{equation}
Therefore, using the matrix lifting technique, we obtain an equivalent reformulation for problem $\mathscr{P}_{\text{RGS}}$ as
\setlength\arraycolsep{2pt}
\begin{eqnarray}
    \mathscr{P}:\mathop{\textrm{minimize}}_{\bs{V},\bs{\lambda}} && \sum_{n,l}\left(\frac{1}{\eta_n}\textrm{Tr}(\bs{V}_{ll}[n,n])+ P_{nl}^{\text{c}}I_{\textrm{Tr}(\bs{V}_{ll}[n,n])\ne 0}\right) \nonumber \\
    \textrm{subject to} && (\ref{constraint:lifted_sdp}),\lambda_k\geq 0, \forall k\in[K]   \label{con:QoS_sdp}\\
    && \sum_{l=1}^{K}\textrm{Tr}(\bs{V}_{ll}[n,n])\leq P_n^{\text{Tx}}, \forall n\in[N] \label{con:transmit_power_sdp}\\
    && \bs{V}\succeq\bs{0}, \textrm{rank}(\bs{V})=1. \label{prob:psd}
\end{eqnarray}
Note that the constraints are still nonconvex due to the nonconvexity of the rank-one constraint. 

\subsection{DC Representations for Rank-One Constraint}
For a positive semidefinite matrix $\bs{V}\in\mathbb{S}_+^{NKL}$, its rank is one if and only if it has only one non-zero singular value, i.e.,
\begin{equation}
  \sigma_i(\bs{V})=0, i=2,\cdots,NKL,
\end{equation}
where $\sigma_i(\bs{V})$ is the $i$-th largest singular value of $\bs{V}$. The trace norm and spectral norm of the positive semidefinite matrix $\bs{V}$ are respectively given as
\begin{equation}
   \textrm{Tr}(\bs{V})=\sum_{i=1}^{NKL}\sigma_i(\bs{V}), \|\bs{V}\| = \sigma_1(\bs{V}).
 \end{equation} 
Thus we obtain an equivalent DC representation for the rank-one constraint of $\bs{V}$:
\begin{equation}\label{eq:DCrepresentation}
  \mathcal{R}(\bs{V})=\textrm{Tr}(\bs{V})-\|\bs{V}\|=0.
\end{equation}
$\mathcal{R}$ is a DC function of $\bs{V}$ since both the trace norm and the spectral norm are convex.

\subsection{Reweighted $\ell_1$ Minimization for Inducing Group Sparsity}
Reweighted $\ell_1$ minimization approach has shown its advantages in enhancing group sparsity for improving the energy-efficiency of cloud radio access networks \cite{dai2016energy,shi2016smoothed}. $\ell_1$-norm is a well recognized convex surrogate for the $\ell_0$-norm. In order to further enhance the sparsity, reweighted $\ell_1$ minimization is proposed to iteratively minimize a weighted $\ell_1$-norm and update the weights. For the objective function of problem $\mathscr{P}$, we observe that the indicator function $I_{\textrm{Tr}(\bs{V}_{ll}[n,n])\ne 0}$ can be interpreted as the $\ell_0$-norm of $\textrm{Tr}(\bs{V}_{ll}[n,n])$. We can thus use the reweighed $\ell_1$ minimization technique via approximating $I_{\textrm{Tr}(\bs{V}_{ll}[n,n])\ne 0}$ by $w_{nl}\textrm{Tr}(\bs{V}_{ll}[n,n])$, which consists of alternatively minimizing the approximated objective function and updating the weight as
\begin{equation}\label{eq:update_weights}
  w_{nl}=\frac{c}{\textrm{Tr}(\bs{V}_{ll}[n,n])+\tau},
\end{equation}
where $\tau>0$ is a constant regularization factor and $c>0$ is a constant.
If $\textrm{Tr}(\bs{V}_{ll}[n,n])$ is small, the reweighted $\ell_1$ minimization approach will put larger weight on the transceiver pair $(n,l)$, which prompts that the inference task $l$ is not preferred to be executed at the $n$-th edge node.

\subsection{Proposed Reweighted Power Minimization Approach}
In this subsection, we provide a reweighted power minimization approach by combining the matrix lifting, DC representation and reweighted $\ell_1$ minimization techniques. In the $j$-th step, we shall update  $\bs{V}^{[j+1]}$ via solving
\begin{eqnarray}
    \mathop{\textrm{minimize}}_{\bs{V},\bs{\lambda}} && \sum_{n,l}\Big(\frac{1}{\eta_n}+ w_{nl}^{[j]}P_{nl}^{\text{c}}\Big)\textrm{Tr}(\bs{V}_{ll}[n,n]) \nonumber \\
    \textrm{subject to} && (\ref{constraint:lifted_sdp}),\lambda_k\geq 0, \forall k\in[K]   \nonumber \\
    && \sum_{l=1}^{K}\textrm{Tr}(\bs{V}_{ll}[n,n])\leq P_n^{\text{Tx}}, \forall n\in[N] \nonumber\\
    && \bs{V}\succeq\bs{0}, \textrm{rank}(\bs{V})=1,\label{reweighted_rank_one}
\end{eqnarray}
and the weights $\{w_{nl}^{[j]}\}$ are updated following (\ref{eq:update_weights}) which are initialized as $1$ at the beginning.

To solve problem (\ref{reweighted_rank_one}) with nonconvex rank-one constraint, we propose to use the DC representation (\ref{eq:DCrepresentation}) by solving the following DC program
\begin{eqnarray}
    \mathscr{P}_{\text{DC}}:\mathop{\textrm{minimize}}_{\bs{V},\bs{\lambda}} && \sum_{n,l}\Big(\frac{1}{\eta_n}+ w_{nl}^{[j]}P_{nl}^{\text{c}}\Big)\textrm{Tr}(\bs{V}_{ll}[n,n])+\mu\mathcal{R}(\bs{V}) \nonumber \\
    \textrm{subject to} && (\ref{constraint:lifted_sdp}),\lambda_k\geq 0, \forall k\in[K]    \nonumber\\
    && \sum_{l=1}^{K}\textrm{Tr}(\bs{V}_{ll}[n,n])\leq P_n^{\text{Tx}}, \forall n\in[N] \nonumber\\
    && \bs{V}\succeq\bs{0},
\end{eqnarray}
where $\mu>0$ is the regularization parameter. Despite of the nonconvexity of the DC problem, problem $\mathscr{P}_{\text{DC}}$ can be efficiently solved by the simplified DC algorithm, i.e., iteratively linearizing the concave part \cite{tao1997convex}. At the $t$-th iteration, we shall solve 
\begin{eqnarray}
    \mathop{\textrm{minimize}}_{\bs{V},\bs{\lambda}} && \sum_{n,l}\Big(\frac{1}{\eta_n}+ w_{nl}^{[j]}P_{nl}^{\text{c}}\Big)\textrm{Tr}(\bs{V}_{ll}[n,n]) \nonumber \\
    &&\quad\quad+\mu(\textrm{Tr}(\bs{V})- \textrm{Tr}(G^{(t)}\bs{V}))\nonumber \\
    \textrm{subject to} && (\ref{constraint:lifted_sdp}),\lambda_k\geq 0, \forall k\in[K]   \nonumber\\
    && \sum_{l=1}^{K}\textrm{Tr}(\bs{V}_{ll}[n,n])\leq P_n^{\text{Tx}}, \forall n\in[N] \nonumber\\
    && \bs{V}\succeq\bs{0}, \label{prob:DCiter}
\end{eqnarray}
where $G^{(t)}$ is one subgradient of spectral norm at $\bs{V}^{(t)}$. It can be computed as $\partial\|\bs{V}\|_2=\bs{u}_1\bs{u}_1^{\sf{H}}$ where $\bs{u}_1$ is the eigenvector corresponding to the largest eigenvalue of matrix $\bs{V}$. This DC algorithm guarantees converging to a stationary point of problem $\mathscr{P}_{\text{DC}}$ from arbitrary initial points \cite{tao1997convex}.

When the reweighted $\ell_1$ minimization algorithm converges at a rank-one solution $\bs{V}^{[j]}$, we can extract the aggregative beamforming vector $\bs{v}^\star$ from the Choleskey decomposition $\bs{V}^{[j]}=\bs{v}^\star{\bs{v}^\star}^{\sf{H}}$. The whole procedure of the proposed reweighted power minimization approach is summarized in Algorithm \ref{algorithm:proposed}.

\SetNlSty{textbf}{}{:}
\IncMargin{1em}
\begin{algorithm}[h]
\textbf{Initialization:} $\bs{V}^{[0]},w_{nl}$.\\
 \While{not converge}{
 $\bs{V}^{(0)} \leftarrow \bs{V}^{[j]}$ \\
 \While{not converge}{
    update $\bs{V}^{(t)}$ as the solution to problem (\ref{prob:DCiter})
 }
 $\bs{V}^{[j+1]} \leftarrow \bs{V}^{(t)}$ \\
 update the weights $\{w_{nl}^{[j+1]}\}$ according to equation (\ref{eq:update_weights})
 }
 obtain $\bs{v}^\star$ through Choleskey decomposition $\bs{V}^{[j]}=\bs{v}^\star{\bs{v}^\star}^{\sf{H}}$. \\
 \textbf{Output:} $\bs{v}^\star$.
 \caption{Proposed Reweighted Power Minimization Approach for Problem $\mathscr{P}$}
 \label{algorithm:proposed}
\end{algorithm}


\section{Numerical Results}\label{sec:simulation}
In this section, we provide numerical experiments for comparing the proposed framework with other state-of-the-art approaches. We generate the edge inference system with $N=4$ APs located at $(\pm400,\pm400)$ meters and $K=4$ mobile users randomly located in the $[-800~800]\times [-800~800]$ meters square region. Each AP is equipped with $L=2$ antennas. The imperfection model of the channel coefficient vector between the $n$-th AP and the $k$-th mobile user is chosen as $\bs{h}_{kn}=10^{-L(d_{kn})/20}(\bs{c}_{kn}+\bs{e}_{kn})$. The path loss model is given by $L(d_{kn})=128.1+37.6\log_{10}{d_{kn}}$, the Rayleigh small scale fading coefficient is given by $\bs{c}_{kn}\sim\mathcal{CN}(\bs{0},\bs{I})$, and the additive error is given by $\bs{e}_{kn}\sim\mathcal{CN}(\bs{0},10^{-4}\bs{I})$. 
As presented in Section \ref{subsec:learninguncertainty}, $D_1$ determines the accuracy of the learned shape of the uncertainty set, while $D_2$ determines the accuracy of the calibrated size of the uncertainty set. To balance these two points, the collected $D$ independent samples of $\bs{h}_{kn}$'s are split evenly for learning the shape and size of the uncertainty ellipsoids, respectively, i.e., $D_1=D_2=D/2$. For each AP, the power amplifier efficiency is chosen as $\eta_1=\cdots=\eta_N=1/4$, the average maximum transmit power is chosen as $P_1^{\text{Tx}}=\cdots=P_N^{\text{Tx}}=1W$, and the computation power consumption for each task $\phi_k$ at the $n$-th AP is chosen as $P_{nk}^{\text{c}}=0.60W$. We set the target SINR as $\gamma_1=\cdots=\gamma_K=\gamma$, the tolerance level as $\epsilon=0.05$, and the confidence level as $\delta=0.05$. The regularization parameters $\tau$ is set as $10^{-6}$ and $\mu$ is set as $10$. 

\subsection{Benefits of Taking CSI Uncertainty into Consideration}\label{subsec:sim_uncertainty}
In this paper, we consider the CSI uncertainty in channel sampling and propose to solve it with a learning-based robust optimization approximation approach. To further reduce the channel sampling overhead, we provide a cost-effective sampling strategy in Secion \ref{subsec:sampling_strategy}. We now evaluate its advantages over the beamformer design without taking the CSI error into consideration by supposing that each task is performed at all APs. Specifically, we collect $D=200$ i.i.d. channel samples in the training phase within one coherent interval for CSI. In the test phase, we only collect one channel sample $\bs{h}^{(1)}$, construct $\bs{H}_k$'s following equation (\ref{eq:efficient_sampling}) and solve the problem
 \begin{eqnarray}
    \mathop{\textrm{minimize}}_{\bs{V},\bs{\lambda}} && \sum_{n,l}\left(\frac{1}{\eta_n}\textrm{Tr}(\bs{V}_{ll}[n,n])+ P_{nl}^{\text{c}}\right) \nonumber \\
    \textrm{subject to} && (\ref{constraint:lifted_sdp}),\lambda_k\geq 0, \forall k\in[K]   \nonumber\\
    && \sum_{l=1}^{K}\textrm{Tr}(\bs{V}_{ll}[n,n])\leq P_n^{\text{Tx}}, \forall n\in[N] \nonumber\\
    && \bs{V}\succeq\bs{0}.
\end{eqnarray}
As comparison, the beamforming design without taking uncertainty into consideration is given by solving the problem
\begin{eqnarray}
    \mathop{\textrm{minimize}}_{\bs{V},\bs{\lambda}} && \sum_{n,l}\left(\frac{1}{\eta_n}\textrm{Tr}(\bs{V}_{ll}[n,n])+ P_{nl}^{\text{c}}\right) \nonumber \\
    \textrm{subject to} && {\bs{h}_{k}^{(1)}}^{\sf{H}}\Big(\frac{1}{\gamma_k}\bs{V}_{kk}-\sum_{l\ne k}\bs{V}_{ll}\Big){\bs{h}_{k}^{(1)}}\geq \sigma_k^2, ~\forall k  \nonumber\\
    && \sum_{l=1}^{K}\textrm{Tr}(\bs{V}_{ll}[n,n])\leq P_n^{\text{Tx}}, ~\forall n, \nonumber\\
    && \bs{V}\succeq\bs{0}.
\end{eqnarray}
Note that we use SDR for both approaches for fairness. We compare two approaches by generating $40000$ realizations of i.i.d. channel samples for testing, and regenerate the training data set for the proposed approach every $200$ realizations. We compute the achieved SINR for each mobile device with the solution to each approach, i.e., $\textrm{SINR}_k(\bs{v};\tilde{\bs{h}})$ where $\tilde{\bs{h}}$ is the true channel coefficient vector, and calculate the number of realizations that the target QoS for each device is met, i.e., $\textrm{SINR}_k\geq \gamma_k$. The results shown in Table \ref{tab:uncertainty} demonstrate that the proposed robust approximation approach has considerably improved the robustness of QoS against CSI errors by a cost-effective sampling approach.

\begin{table}[h]
\centering
\caption{Number of tests that QoS is met}
        \label{tab:uncertainty}
        \begin{tabular}{c|c c c c}
        User Index  & 1& 2& 3& 4 \\\hline
        Proposed Approach  & 39946 & 39946 & 39946 & 39946\\
        Without considering uncertainty  & 15205 & 15123 & 15197 & 15214
        \end{tabular}
\end{table}

\subsection{Overcoming the Over-Conservativeness of Scenario Generation}\label{subsec:SG}
As we point out in Section \ref{subsec:problem_analysis}, the scenario generation approach is over-conservative since it imposes that the target QoS constraints are satisfied for all samples, which would lead to a smaller feasible region. Here we use numerical experiments to demonstrate the advantage of the presented robust optimization approximation approach in overcoming the over-conservativeness. Consider the feasibility problem of the robust optimization approximation approach given by
\begin{eqnarray}
    \mathop{\textrm{find}} &&  \bs{V},\bs{\lambda} \nonumber \\
    \textrm{subject to} && (\ref{constraint:lifted_sdp}),\lambda_k\geq 0, \forall k\in[K],   \nonumber\\
    && \sum_{l=1}^{K}\textrm{Tr}(\bs{V}_{ll}[n,n])\leq P_n^{\text{Tx}}, \forall n\in[N], \nonumber\\
    && \bs{V}\succeq\bs{0},
\end{eqnarray}
and the feasibility problem of the scenario approach given by
\begin{eqnarray}
    \mathop{\textrm{find}} &&  \bs{V} \nonumber \\
    \textrm{subject to} && {\bs{h}_{k}^{(i)}}^{\sf{H}}\Big(\frac{1}{\gamma_k}\bs{V}_{kk}-\sum_{l\ne k}\bs{V}_{ll}\Big){\bs{h}_{k}^{(i)}}\geq \sigma_k^2, ~\forall k, i  \nonumber\\
    && \sum_{l=1}^{K}\textrm{Tr}(\bs{V}_{ll}[n,n])\leq P_n^{\text{Tx}}, ~\forall n, \nonumber\\
    && \bs{V}\succeq\bs{0}.
\end{eqnarray}
Note that we adopt the SDR technique in both approach for purpose of fairness. We collect $D=200$ i.i.d. channel samples for each realization, run both algorithms for $25$ random realizations of the data set, and compare the probability of yielding feasible solutions using the scenario generation approach and the presented robust optimization approximation approach. The results in Fig. \ref{fig:feasibility} reveal that the statistical learning based robust approximation considerably improves the probability of feasibility compared with the scenario generation approach though we only obtain sufficient conditions for the robust optimization counterpart using S-procedure in Section \ref{subsec:tractable_reformulation}.

\begin{figure}[h]
        \centering
        \includegraphics[width=\columnwidth]{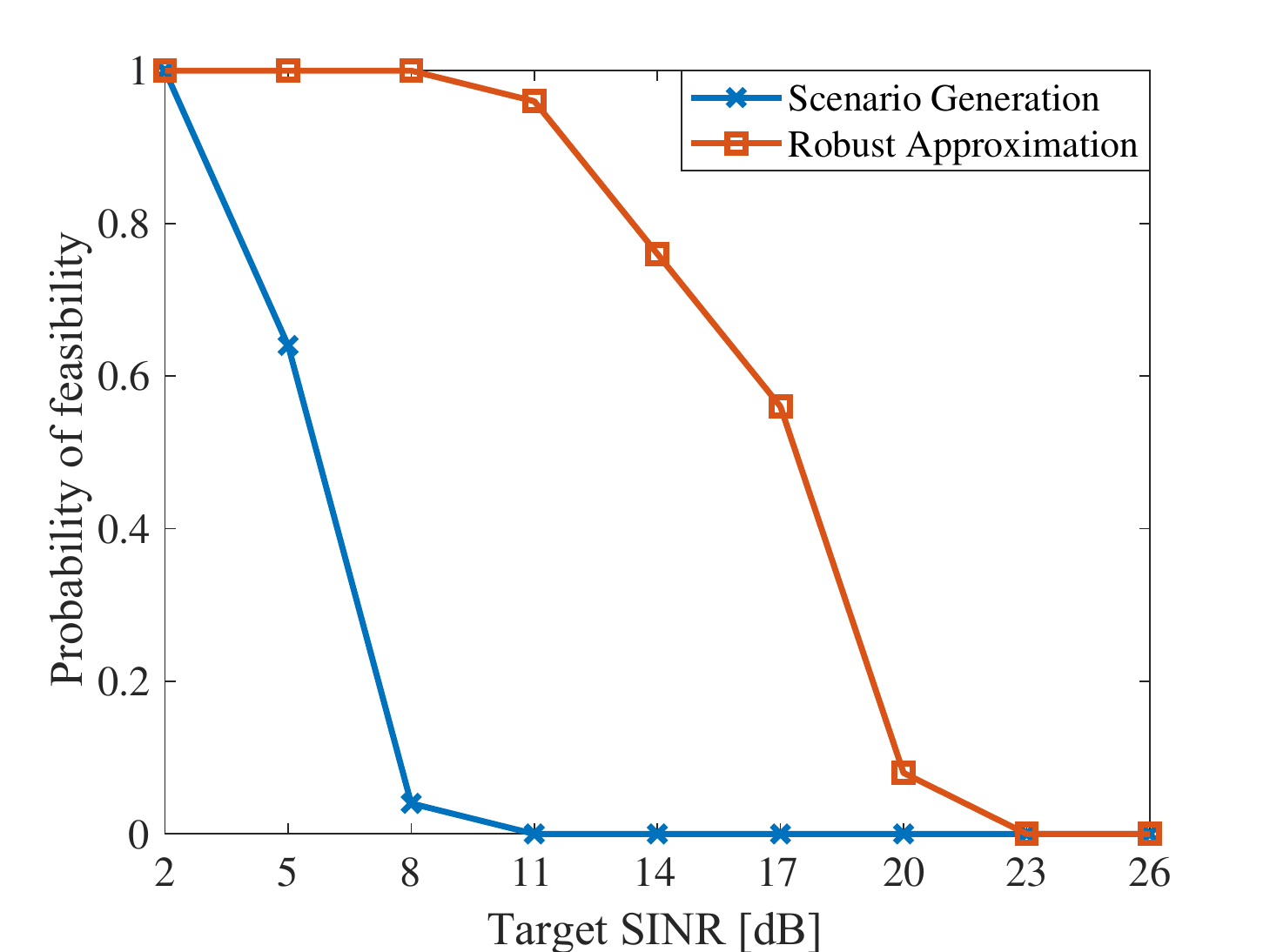}
        \caption{Probability of feasibility using scenario generation and the robust optimization approximation approach over the target SINR $\gamma$.}\label{fig:feasibility}
\end{figure}

\subsection{Convergence Behavior}\label{subsec:convergence}
By choosing the reweighting parameter as $c=1/\ln(1+\tau^{-1})$, the proposed reweighted power minimization approach, i.e., Algorithm \ref{algorithm:proposed}, essentially approximates the $\ell_0$-norm according to
$I_{x\ne 0} = \|x\|_0 = \lim_{\tau \rightarrow 0}{\ln(1+x\tau^{-1})}/{\ln(1+\tau^{-1})}$,
and minimizes the approximated objective function
\begin{align}
  f(\bs{V})&=\sum_{n,l}\Big(\frac{1}{\eta_n}\textrm{Tr}(\bs{V}_{ll}[n,n])\nonumber\\
  &\quad+ P_{nl}^{\text{c}}\frac{\ln(1+\tau^{-1}\textrm{Tr}(\bs{V}_{ll}[n,n]))}{\ln(1+\tau^{-1})}\Big)+\mu\mathcal{R}(\bs{V})\label{eq:obj_value}
\end{align}
under constraints (\ref{con:QoS_sdp}) and (\ref{con:transmit_power_sdp}) using an majorization-minimization (MM) technique as stated in \cite{dai2016energy}. Fig. \ref{fig:convergence_obj} illustrates the convergence behavior of the proposed reweighted power minimization approach in terms of the objective function $f$ by collecting $D=200$ channel samples. We also plot the corresponding trajectories of the group sparsity of the aggregative beamforming vector $\bs{v}$ in Fig. \ref{fig:convergence_num}, i.e., total number of inference tasks performed at all edge computing nodes. We observe that the number of tasks to be performed at edge computing nodes increases with a greater value of target QoS $\gamma$, which leads to higher total power consumption of the edge inference system.
\begin{figure}[h]
        \centering
        \includegraphics[width=\columnwidth]{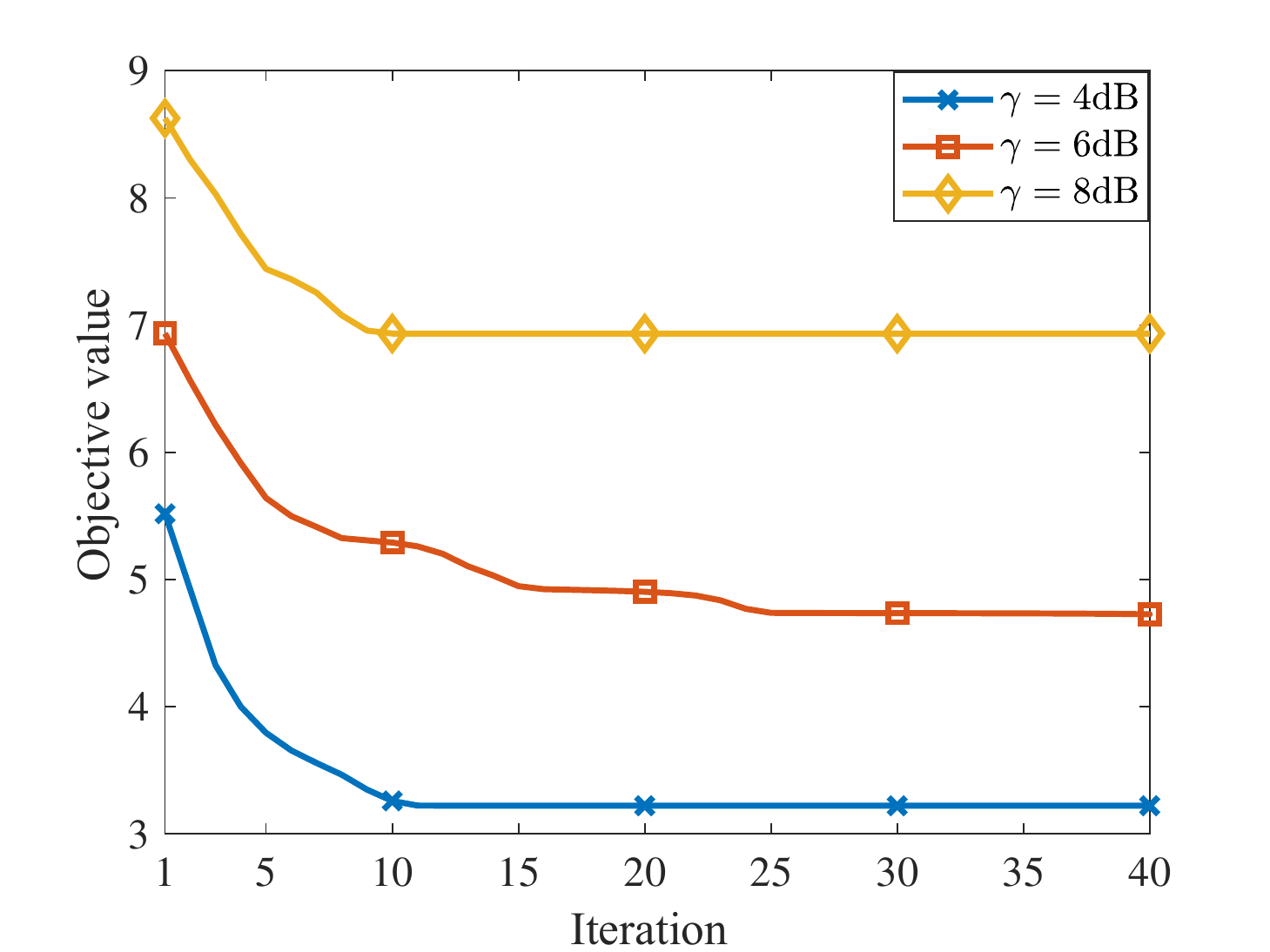}
        \caption{Convergence behavior of the proposed reweighted power minimization approach with different target SINR $\gamma$.}\label{fig:convergence_obj}
\end{figure}

\begin{figure}[h]
        \centering
        \includegraphics[width=\columnwidth]{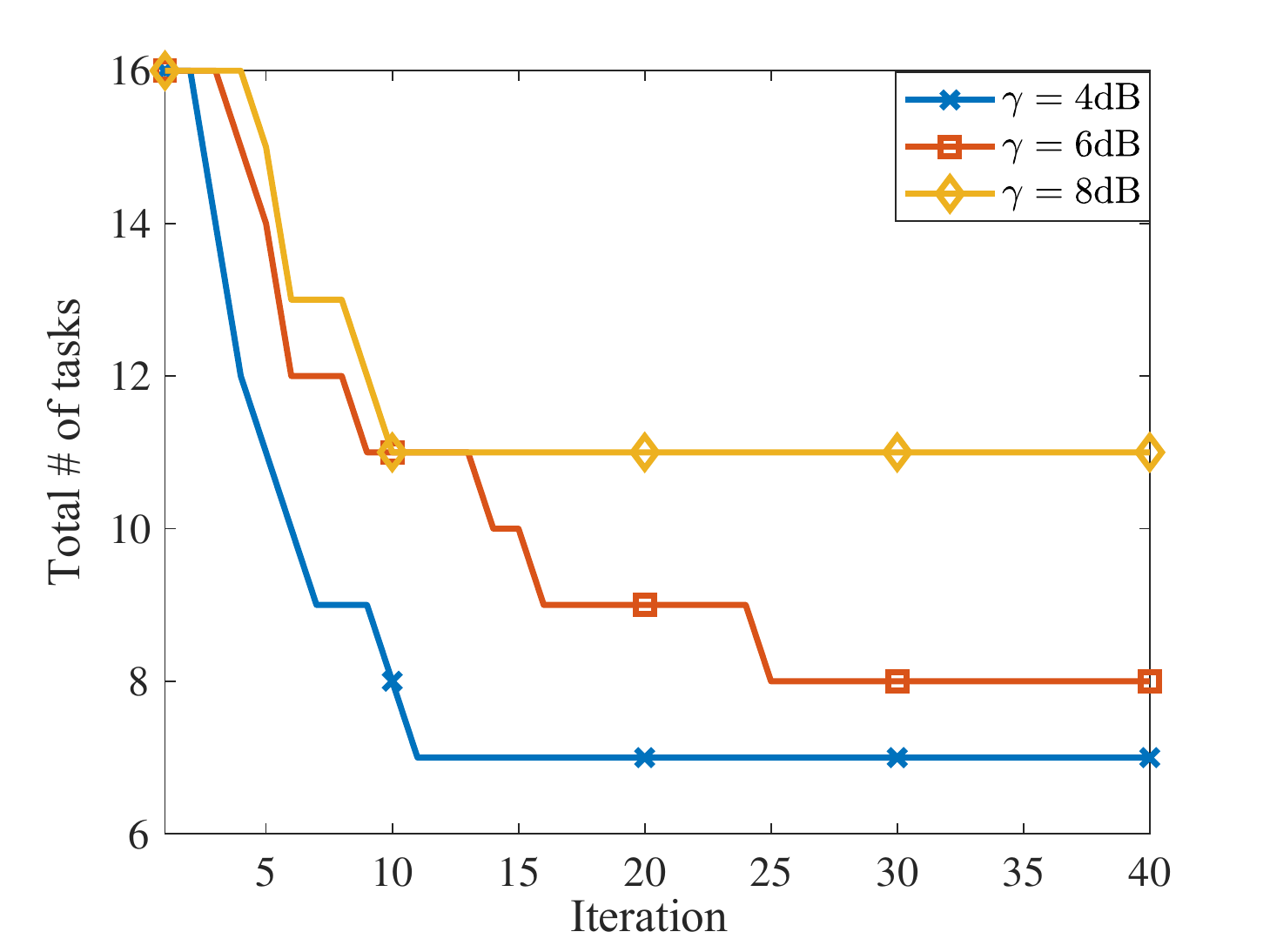}
        \caption{Trajectories of the total number of inference tasks performed at all edge computing nodes with different target SINR $\gamma$.}\label{fig:convergence_num}
\end{figure}

\subsection{Total Power Consumptions over Target SINR}
We then conduct numerical results to compare the performance of different algorithms for problem $\mathscr{P}$ with $D=200$ i.i.d. channel samples, including the proposed reweighted power minimization approach termed as ``reweighted+DC'' and other state-of-the-art algorithms listed below: 
\begin{itemize}
    \item ``mixed $\ell_1/\ell_2$+SDR'': This algorithm is proposed in \cite{shi2015robust_TSP}, which adopts the quadratic variational form of the weighted mixed $\ell_1/\ell_2$-norm for inducing group sparsity and SDR to address the nonconvex quadratic constraints.
    \item ``reweighted+SDR'': To improve the energy efficiency of downlink transmission in cloud-RAN, we adopt the iteratively reweighted minimization algorithm \cite{dai2016energy} for inducing the group sparsity and SDR \cite{luo2007approximation} for the nonconvex quadratic constraints.
    \item ``CB+SDR'': We assume that all tasks are performed at each AP and conduct coordinated beamforming for minimizing the transmission power consumption under probabilistic-QoS constraints.
\end{itemize}

We also set $c=1/\ln(1+\tau^{-1})$ as stated in Sec. \ref{subsec:convergence}. The performances of all algorithms averaged over $100$ channel realizations are illustrated in Fig. \ref{fig:totalpower} and Fig. \ref{fig:num_tasks}. Fig. \ref{fig:totalpower} presents the total power consumption of each algorithm and demonstrates that the proposed DC algorithm yields lower total power consumption than other approaches, which is owed to its better capability to induce group sparsity as shown in Fig. \ref{fig:num_tasks}. 
Note that the total number of tasks for the ``CB+SDR'' algorithm is always $KN=16$.

\begin{figure}[h]
        \centering
        \includegraphics[width=\columnwidth]{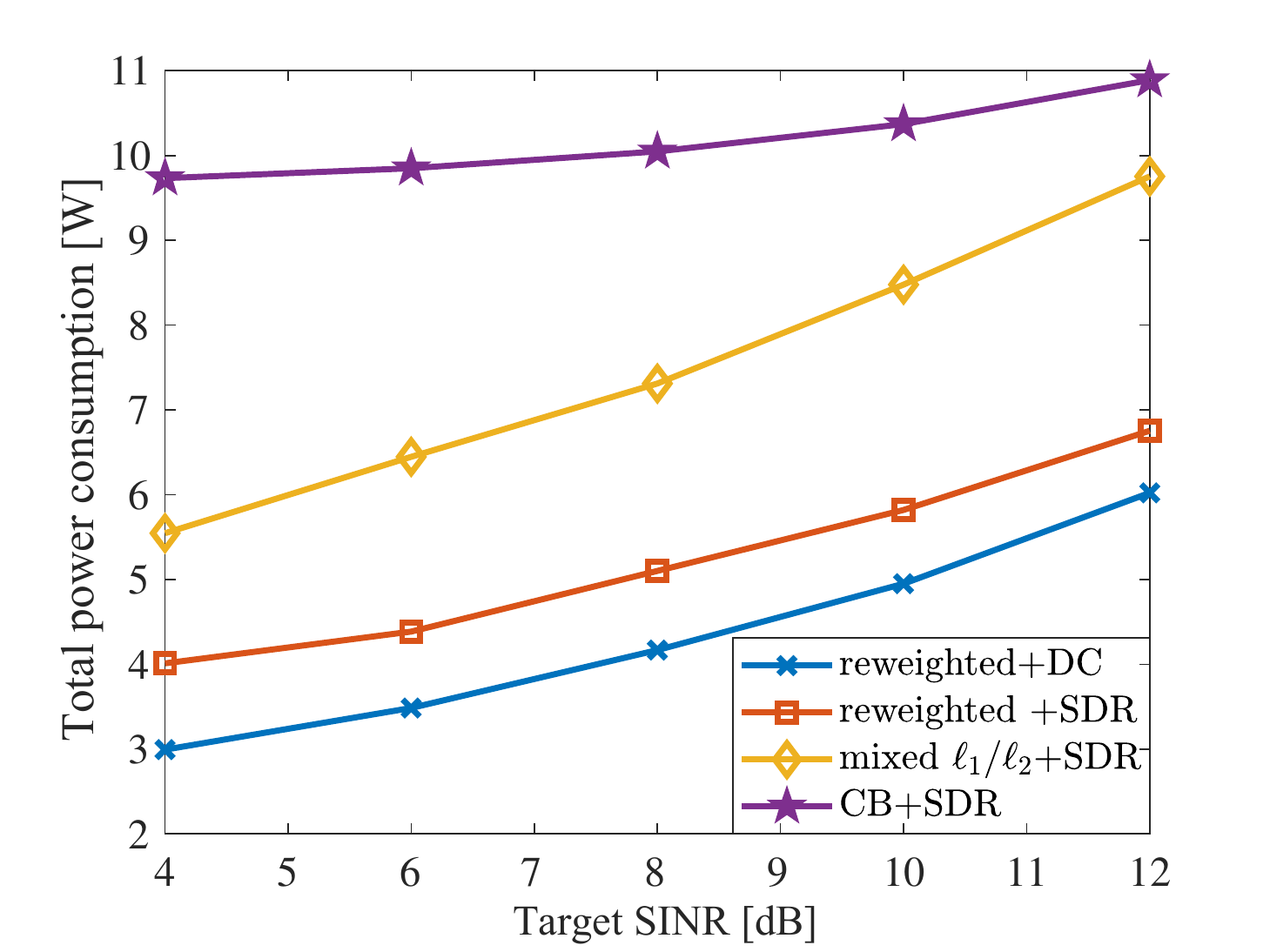}
        \caption{Total power consumption over target SINR.}\label{fig:totalpower}
\end{figure}


\begin{figure}[h]
        \centering
        \includegraphics[width=\columnwidth]{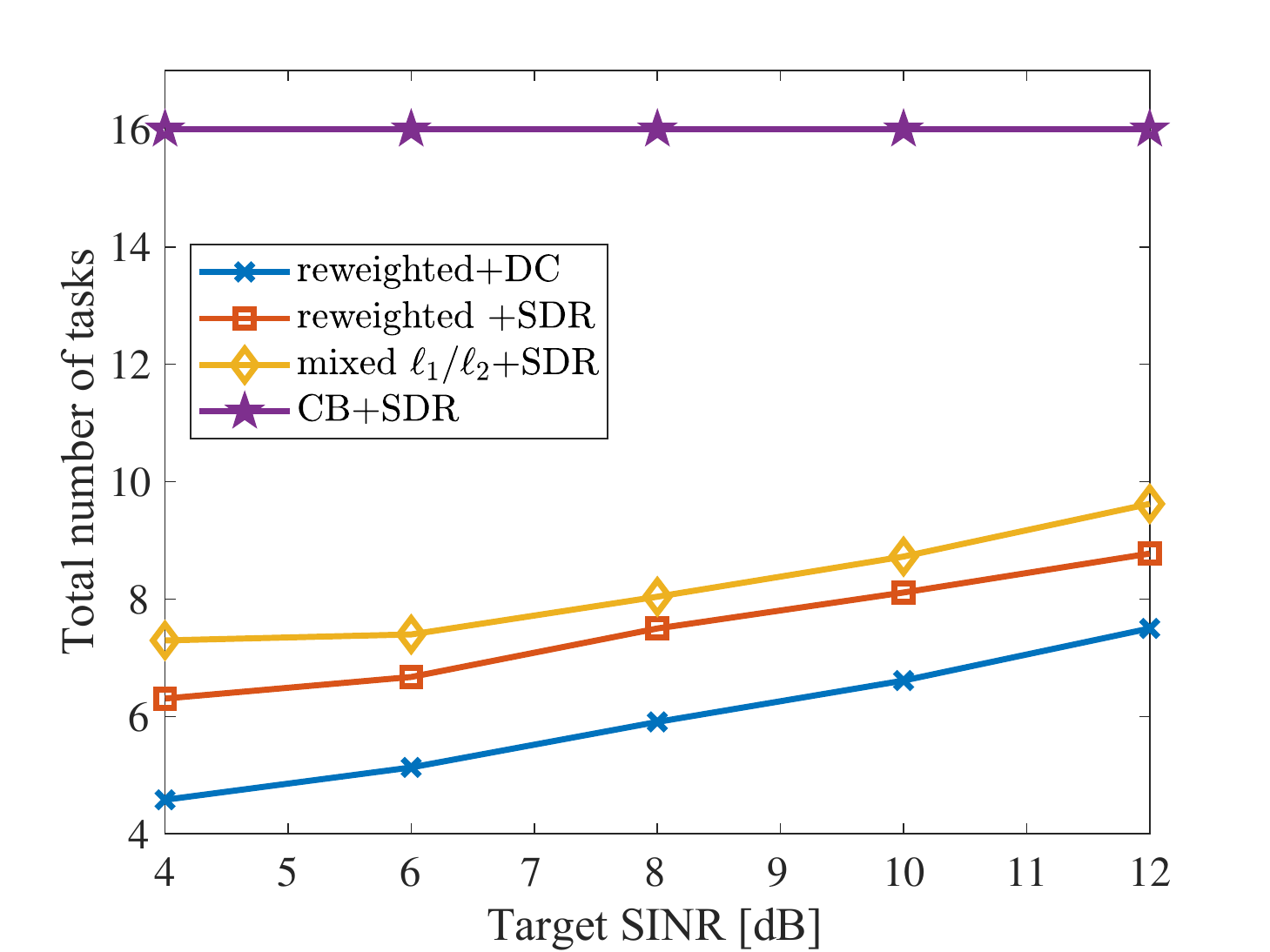}
        \caption{Total \# of tasks performed at APs over target SINR.}\label{fig:num_tasks}
\end{figure}

Through all numerical results, we have seen considerable advantages of the presented statistical learning based robust optimization approximation and 
the proposed reweighted power minimization algorithm in providing energy-efficient processing and robust transmission service for edge inference.

\section{Conclusion}
In this work, we presented an energy-efficient processing and robust cooperative transmission framework for executing deep learning inference tasks for mobile devices. Specifically, we proposed to minimize the sum of computation power and transmission power consumption under the probabilistic-QoS constraints via adaptive task selection and coordinated beamforming design. The joint chance constraints therein were further addressed by a statistical learning based robust optimization approximation approach. This yielded a
group sparse beamforming problem with nonconvex quadratic constraints. We
then developed a reweighted power minimization approach by iteratively solving a DC regularized reweighted $\ell_1$ minimization problem and updating the weights, thereby tackling both the group sparse
objective function and nonconvex quadratic constraints. Numerical results
demonstrated that the proposed approach achieved the lowest total power consumption among other state-of-the-art algorithms, and avoided the drawbacks of other methods for joint chance-constrained programs.

There are still some open problems to be studied:
\begin{itemize}
\item This work considers the architecture that each inference task is performed at multiple base stations separately. An interesting problem is to consider the hierarchical distributed structure of deep neural networks over the cloud and the edge \cite{teerapittayanon2017distributed}.
  \item In this work, we consider a basic ellipsoid model for each uncertain channel coefficient vector. It is interesting to use data-driven approach with more complicated model of the high probability region to reduce its volume, such as clustering the data samples and using a union of ellipsoids as the high probability region. 
  \item It is still an open problem to provide the theoretical guarantee of the proposed reweighted power minimization algorithm since the conditions for convergence guarantee of reweighted approach in \cite{dai2016energy,Yuanming_cvxsmooth18} are not met.
\end{itemize}
\appendices
\section{Derivation of (\ref{constraint:nonconvex_sdp}) Using S-Procedure}\label{append:Sprocedure}
We first rewrite (\ref{eq:uncertainty1}) as
\begin{equation}
  \bs{h}_{k}\tau_k = \hat{\bs{h}}_{k}\tau_k+ \bs{B}_{k}\tilde{\bs{u}}_{k}, \tilde{\bs{u}}_{k}^{\sf{H}}\tilde{\bs{u}}_{k}\leq \tau_{k}^2,
\end{equation}
where $\bs{u}_{k}=\tilde{\bs{u}}_{k}/\tau_{k}\in\mathbb{C}^{L},\tau_k>0$. Let
\begin{equation}
  \bs{x}_{k}=\begin{bmatrix}
    \tau_{k}^{\sf{H}} & \tilde{\bs{u}}_{k}^{\sf{H}}
\end{bmatrix}^{\sf{H}}\in\mathbb{C}^{NL+1},
\end{equation}
we can obtain that 
\begin{equation}
  \bs{h}_k\tau_k = \bs{H}_{k}\bs{x}_k.
\end{equation}
Thus we know
\begin{align}
  &\bs{h}_k^{\sf{H}}(\frac{1}{\gamma_k}\bs{v}_k\bs{v}_k^{\sf{H}}-\sum_{l\ne k} \bs{v}_l\bs{v}_l^{\sf{H}})\bs{h}_k- \sigma_k^2\geq 0 \\
  \Leftrightarrow&(\bs{h}_k\tau_k)^{\sf{H}}(\frac{1}{\gamma_k}\bs{v}_k\bs{v}_k^{\sf{H}}-\sum_{l\ne k} \bs{v}_l\bs{v}_l^{\sf{H}})\bs{h}_k\tau_k- \sigma_k^2\tau_k^2\geq 0 \\
  \Leftrightarrow& (\bs{H}_{k}\bs{x}_k)^{\sf{H}}(\frac{1}{\gamma_k}\bs{v}_k\bs{v}_k^{\sf{H}}-\sum_{l\ne k} \bs{v}_l\bs{v}_l^{\sf{H}})\bs{H}_{k}\bs{x}_k- \sigma_k^2\tau_k^2\geq 0 \\
  \Leftrightarrow&\bs{x}_{k}^{\sf{H}}\bs{P}_k^{0}\bs{x}_{k}\geq 0,
\end{align}
where $\bs{P}_k^{0}\in\mathbb{S}^{NL+1}$ is given by
\begin{equation}
  \bs{H}_{k}^{\sf{H}}(\frac{1}{\gamma_k}\bs{v}_k\bs{v}_k^{\sf{H}}-\sum_{l\ne k} \bs{v}_l\bs{v}_l^{\sf{H}})\bs{H}_{k}-\begin{bmatrix}
    \sigma_k^2 & 0 & \cdots & 0 \\
    0 & 0 & \cdots & 0 \\
    \vdots & \vdots & \ddots & \vdots \\
    0 & 0 & \cdots & 0
  \end{bmatrix}.
\end{equation}
Likewise, $\tilde{\bs{u}}_{k}^{\sf{H}}\tilde{\bs{u}}_{k}\leq \tau_{k}^2,$ can be rewritten as
\begin{equation}
  \bs{x}_{k}^{\sf{H}}\bs{P}_k^{1}\bs{x}_{k}\geq 0,
\end{equation}
where $\bs{P}_k^{1}\in\mathbb{S}^{NL+1}$ is given by
\begin{equation}
  \bs{P}_k^{1}=\begin{bmatrix}
      1 &  \\
      & -\bs{I}_N
  \end{bmatrix}
\end{equation}
Thus, we shall use the S-procedure
\begin{equation}
  \bs{x}_{k}^{\sf{H}}\bs{P}_k^{1}\bs{x}_{k}\geq 0  \Longrightarrow \bs{x}_{k}^{\sf{H}}\bs{P}_k^{0}\bs{x}_{k}\geq 0,
\end{equation}
which is given by
\begin{equation}
  \bs{P}_k^{0}\geq \lambda_{k}\bs{P}_k^{1}, \lambda_{k}\geq 0.
\end{equation}
Therefore, we obtain the tractable reformulation for the joint chance constraints (\ref{constraint:probQoS}) as
\begin{align}
  &\bs{H}_{k}^{\sf{H}}(\frac{1}{\gamma_k}\bs{v}_{k}\bs{v}_{k}^{\sf{H}}-\sum_{l\ne k}\bs{v}_{l}\bs{v}_{l}^{\sf{H}})\bs{H}_{k}\succeq \bs{Q}_k,
\end{align}
where $\bs{\lambda}=[\bs{\lambda}_{1},\cdots,\bs{\lambda}_{K}]=[\lambda_{nk}]\in\mathbb{R}_{+}^{N\times K}$ and $\bs{Q}_k$ is given by
\begin{equation}
  \bs{Q}_k=\begin{bmatrix}
    \lambda_{k}+\sigma_k^2 & \\
    & -\lambda_{k}\bs{I}_{NL} 
  \end{bmatrix}\in\mathbb{C}^{(NL+1)\times(NL+1)}.
\end{equation}

\bibliographystyle{IEEEtran}
\bibliography{reliable_edge_processing}

\end{document}